\documentclass[12pt]{article}
\usepackage{amssymb}
\usepackage{amsmath,bm}
\usepackage{graphics}
\usepackage{graphicx,epsfig}
\usepackage{subfigure}
\usepackage{placeins}
\usepackage{indentfirst}
\usepackage{setspace}
\usepackage{enumerate}
\usepackage{enumitem}
\usepackage{natbib}
\usepackage{caption}
\usepackage{varioref}
\usepackage{booktabs,tabularx}
\usepackage{color}
\usepackage{scalefnt}

\usepackage{epstopdf}
\usepackage{hyperref,amsfonts,dsfont}
\usepackage[title]{appendix}
\usepackage{amsthm}
\makeatletter \oddsidemargin  -.1in \evensidemargin -.1in
	\newtheorem{remark}{Remark}[section]
	\newtheorem{proposition}{Proposition}[section]
	\numberwithin{equation}{section}
	\renewcommand{\theequation}{\thesection.\arabic{equation}}
 
\begin{document}
	\newcommand{\bea}{\begin{eqnarray}}
		\newcommand{\eea}{\end{eqnarray}}
	\newcommand{\nn}{\nonumber}
	\newcommand{\bee}{\begin{eqnarray*}}
		\newcommand{\eee}{\end{eqnarray*}}
	\newcommand{\lb}{\label}
	\newcommand{\nii}{\noindent}
	\newcommand{\ii}{\indent}
	\numberwithin{equation}{section}
	\renewcommand{\theequation}{\thesection.\arabic{equation}}
     \onehalfspacing
	\title{\bf A novel approach to generate distributions {\color{black} with applications to regression modeling}}
	\author{ Subhankar {Dutta}$^1$,~  Roberto {Vila}$^{2\,\!}$\thanks {Email address: rovig161@gmail.com}~~and~~Terezinha K. A. {Ribeiro}$^2$
		\\{\it \small $^1$Department of Mathematics, Bioinformatics and Computer Applications,}\\[-0,2cm]
        {\it \small Maulana Azad National Institute of Technology, Bhopal, 462003, India}\\
        {\it \small $^2$Department of Statistics, University of Brasilia, Brasilia/DF, 70910-900, Brazil}}
\date{}
\maketitle
\begin{center}
	\textbf{Abstract}
\end{center}
A novel approach to adding {\color{black} an} additional parameter to a family of distributions for better adaptability has been put forth.
This approach yields a versatile class of distributions supported on the positive real line. 
{\color{black} An important advantage of the proposed family is that the additional parameter admits a clear interpretation in terms of tail behavior, providing a simple mechanism for modulating tail heaviness.}
We proceed to analyze its mathematical characteristics, such as critical points, modality, stochastic representation, identifiability, quantiles, moments, and truncated moments.
We present {\color{black} two new regression models} for {\color{black} positive} continuous data based on submodels of the newly proposed family of distributions, in which the distribution of the response variable is reparameterized in terms of the median. We use the maximum likelihood method to estimate the parameters, which was implemented through the 
{\tt gamlss} package in {\tt R}. {\color{black} The proposed regression models were applied to a real dataset, and their advantages over common alternative regression models were demonstrated through quantile residual analysis and information criteria.}



\section{Introduction}
\noindent Experts are becoming more interested in creating novel lifetime distributions to get around complex models as the variety of real-world problems along with applications encompassing complex phenomena grows. As a result, substantial improvements have been achieved in creating a variety of novel distribution classes that will produce more adaptable distributions rather than the traditional ones and offer improved information for modeling. The very first individual to propose a novel distribution by using a baseline probability was \cite{kumaraswamy1980generalized}. After that \cite{azzalini1985class} followed with a method for creating novel distributions with modifying a symmetric distribution through the addition of a skewing factor. There are five alternative ways to generate an entirely novel class of distributions. The initial method involves applying differential equations, and the subsequent one involves creating a weighted version of the baseline probability.  The third concept is about supplementing the baseline probability with a number of extra parameters required.  The discretization of the continuous probability densities is the subsequent strategy. The final strategy is the distribution transforming technique that creates a probabilistic representations of the baseline probability to alter the existing probability distribution models. The resulting connection ought to satisfy the distribution theory postulates.

The final approach for developing an entirely novel type of probability distribution is of relevance to us in this current piece of work. Various generators depending on various quantitative functional connections have been suggested in the existing literature. In this regard, \cite{mudholkar1993exponentiated} used a positive parameter to exponentiate an existing baseline distribution in order to construct the exponentiated class of probability distributions.  \cite{marshall1997new} added an extra shape parameter to the transforming technique in order to apply it to the survival function. \cite{eugene2002beta} developed the Beta-G class of distributions by using beta distribution as a generating function. \cite{zografos2009families} proposed novel family of distributions using gamma distribution. \cite{tahir2016logistic} suggested a transmuted family of distribution using logistic distribution. \cite{alizadeh2017generalized} introduced a generalized family of transmuted distributions. \cite{aslam2018cubic} proposed a cubic transmuted family of distributions. 

\cite{alzaatreh2013new} recently put up a technique for creating continuous distribution families, known as T-X family of distributions. \cite{maurya2016new} proposed log transformed family of distributions. Further, \cite{kumar2015method} introduced DUS family of distributions. Some researchers established some family of distributions using trigonometric functions as well. For instance,  \cite{kumar2015new} introduced Sine-g family of distributions. \cite{bakouch2018new} proposed new exponential trigonometric function to generate family of distributions.  \cite{souza2019general} introduced cosine-G class of distributions. In statistical literature, to find a novel generator of a family of distributions, numerous authors have made extensive modifications. For more details see  \cite{nassar2019marshall}, \cite{aslam2020modified}, \cite{el2019odd}, \cite{chakraborty2022kumaraswamy}, \cite{alsolmi2023investigating},  \cite{atchade2023topp} and  \cite{makubate2024novel}. 

{\color{black} The purpose of this study is to increase the flexibility of a baseline distribution through the introduction of an additional parameter. We refer to the resulting transformation as the Dutta transformation (DT) method. The proposed DT can be applied to any baseline distribution, generating a wide variety of density and hazard rate shapes. Moreover, the additional parameter admits a clear interpretation in terms of tail behavior, acting as a tail-modulation parameter that preserves the asymptotic tail behavior of the baseline distribution when $\beta>1$ and produces heavier tails when $0<\beta<1$. Thus, it provides an explicit mechanism for controlling tail heaviness while preserving the overall structure of the baseline model. Owing to its simple form, the DT is easy to implement in practice and yields flexible models that may serve as competitive alternatives to existing distributions.
}

In many fields, random phenomena are influenced by various factors, making it reasonable to assume that the data observed by the experimenter may not be identically distributed. One approach to handling such scenarios is to employ regression frameworks that model the parameters of the distribution of random phenomena (response variable) as functions of the covariates. To extend and demonstrate the applicability of a submodel generated by the DT approach, we propose a new regression model for unimodal continuous data. In this model, the regression structure is assigned to the distribution parameters, one of which represents the median.
This approach offers flexibility and allows for a direct interpretation of the relationships between response and covariates. The median is a measure of central tendency that is easy to interpret, robust to outliers, and more representative of asymmetric data.  Median regression models have attracted considerable attention over the years. A few examples are as follows. The median regression is a special case of regression quantiles proposed by \cite{koenker1978regression}. Median regression models for estimating how covariates affect the median survival time were proposed by \cite{ying1995survival}. Median-dispersion reparameterizations of the Kumaraswamy distribution, which facilitate its use in regression models, were introduced by \cite{mitnik2013kumaraswamy}.
A flexible class of models called power logit regression models, indexed by three parameters—one of which represents the median—was developed by \cite{queiroz2024power}.  \cite{bourguignon2025parametric} proposed a quantile regression model based on the inverse Gaussian distribution, discussing the special case of median regression and comparing it with the classical mean regression model. {\color{black}  \cite{otiniano2026regression} proposes a new class of regression models for heterogeneous extreme data based on the bimodal GEV distribution, where the median of the response is modeled through covariates.}
Additionally, several probability distributions used in the generalized additive model for location, scale and shape (GAMLSS)  \citep{stasinopoulos2008generalized, stasinopoulos2017flexible} are parameterized in terms of the median, facilitating the interpretation of regression coefficients.

The study is broken down in the following manner. In Section \ref{DTG}, the new generator called Dutta transformed-G (DT-G) family has been proposed.   
In Section \ref{properties}, we investigate several mathematical properties of the proposed DT-G model, including critical points, modality, stochastic representation, identifiability, quantiles, moments, and truncated moments.
{\color{black} In Section \ref{DTReg}, two new distributions, which are specific cases of the DT approach called the Dutta transformed exponential distribution (DTED) and Dutta transformed Weibull distribution (DTWD), are presented.
Additionally, we introduce alternative parameterizations of the DTED and DTWD that allow one of their parameters to represent the median. We then develop two new regression models based on the DTED and DTWD and discuss parameter estimation and diagnostic methods. In Section \ref{App}, we present an application to real data that illustrates the usefulness of the proposed regression models compared to alternative models commonly used in the literature.} Finally, in Section \ref{CR}, we provide some concluding remarks.  {\color{black} The code for all computational procedures in this article is available in an open-access repository at \url{https://github.com/terezinharibeiro/DT_Regression_Models}.}

\section{Proposed DT-G model}\label{DTG}
Similarly, with the goal to offer an entirely novel type of parametric lifetime distribution, the present paper will recommend a family of distributions, called the Dutta family (DF) of distributions. For any arbitrary baseline distribution function $G(x)$, the cumulative distribution function (CDF) and probability density function (PDF) of DF distributions can be expressed as 
\begin{align}\label{cdf-main}
    F(x)=G(x) e^{-\bar{G}(x)^{\beta}},~~ x>0;\beta>0, 
\end{align}
and
\begin{align*}
    f(x)= g(x)\left[1+\beta G(x)\bar{G}(x)^{\beta-1}\right]e^{-\bar{G}(x)^{\beta}}, ~~ {\color{black} x>0;\beta>0,} 
\end{align*}
respectively, where $g(x)$ is the pdf of the distribution $G(x)$, $\bar{G}(x)=1-G(x)$ and $\beta$ is the shape parameter {\color{black} with a direct interpretation in terms of tail behavior}. 
A random variable with DT-G distribution will be called DT-G random variable.
It is clear that $\lim_{x\to 0^+}f(x)=\lim_{x\to \infty}f(x)=0$. The survival function (SF) and hazard rate function (HRF) can be obtained as 
\begin{align*}
    S(t)= 1-G(t) e^{-\bar{G}(t)^{\beta}};~~t>0;{\color{black} \beta>0,} 
\end{align*}
and 
\begin{align*}
    h(t)=\frac{g(t)\left[1+\beta G(t)\bar{G}(t)^{\beta-1}\right]e^{-\bar{G}(t)^{\beta}}}{1-G(t) e^{-\bar{G}(t)^{\beta}}},~~{\color{black} t>0;\beta>0,} 
\end{align*}
 respectively.   

{\color{black}
An important advantage of the proposed family \eqref{cdf-main} is the explicit interpretation of the parameter $\beta$ in terms of tail behavior. Indeed,
\[
\frac{S(x)}{\bar G(x)}
=
\frac{1-e^{-\bar G(x)^\beta}}{\bar G(x)^\beta}
\,\bar G(x)^{\beta-1}
+
e^{-\bar G(x)^\beta}.
\]
Moreover,
\[
\lim_{x\to\infty}
\frac{1-e^{-\bar G(x)^\beta}}{\bar G(x)^\beta}
=
\lim_{x\to\infty}
e^{-\bar G(x)^\beta}
=
1
\quad
\text{and}
\quad
\lim_{x\to\infty}
\bar G(x)^{\beta-1}
=
\begin{cases}
0, & \beta>1,\\[2mm]
1, & \beta=1,\\[2mm]
\infty, & 0<\beta<1.
\end{cases}
\]
Therefore,
\[
\lim_{x\to\infty}
\frac{S(x)}{\bar G(x)}
=
\begin{cases}
1, & \beta>1,\\[2mm]
2, & \beta=1,\\[2mm]
\infty, & 0<\beta<1.
\end{cases}
\]

Hence, for $\beta>1$, the proposed transformation preserves the asymptotic tail behavior of the baseline distribution, implying that the resulting and baseline distributions are tail-equivalent. For $\beta=1$, the survival function is asymptotically twice that of the baseline model. In contrast, when $0<\beta<1$, the ratio $S(x)/\bar G(x)$ diverges to infinity, showing that the transformed distribution possesses a heavier tail than the baseline distribution.

This result provides a universal asymptotic characterization of the role of $\beta$, valid for any baseline distribution $G$. Unlike many existing generated families, whose tail properties often depend on the specific baseline model and therefore require separate analyses, the proposed family admits a transparent tail interpretation. Consequently, the parameter $\beta$ offers an explicit mechanism for modulating tail heaviness while preserving the overall structure of the baseline distribution.
}

 {\color{black}
 \begin{remark}
Note that the distribution defined in \eqref{cdf-main} may also be viewed within the general framework of generated distributions obtained by composing a baseline CDF $G$ with a monotone mapping on $[0,1]$. In particular, the proposed model corresponds to the transformation
$
x \mapsto x\exp\{-(1-x)^\beta\}
$, $0\leqslant x\leqslant 1$.
This construction is related to Lehmann-type models \citep{10.1214/aoms/1177729080}, although the classical Lehmann alternatives  are primarily based on power transformations of the form $G^\alpha$, originally introduced in connection with order statistics and rank-based procedures. Thus, while the proposed family shares the same transformation principle, it does not belong to the original Lehmann class in its strict sense.
Moreover, in our construction, the inclusion of the parameter $\beta$ not only increases the flexibility of the parent distribution but also provides an explicit mechanism for modulating tail behavior.
\end{remark}
 }

\section{Some structural properties}\label{properties}

In this section some mathematical properties associated with the proposed DT-G model are investigated.
{\color{black} Such results are standard when introducing new 
distributions, as they establish the theoretical and practical foundations 
for statistical use. In particular, they clarify the model's shape and 
flexibility, enable simulation, ensure identifiability, and provide tools 
(quantiles and moments) for inference, diagnostics, and applications. 
Hence, these properties make the model suitable for both theoretical study 
and practical modelling.}

\subsection{Critical points and modality}

A simple calculation shows that
\begin{align*}
    [\log(f(x))]'
    =
    {g'(x)\over g(x)}
+   
\beta g(x)[1-G(x)]^{\beta-2}
\left\{
V(x)
+
1-G(x)
\right\},
\end{align*}
where
\begin{align*}
    V(x)
    \equiv {
1-\beta G(x)
\over 
1+\beta G(x)[1-G(x)]^{\beta-1}
}.
\end{align*}
Therefore, $x$ is a critical point for $f$ if and only if $x$ satisfies the following equation
    
\begin{align}\label{crit-point}
 U(x)
        =
\beta [1-G(x)]^{\beta-2}
\left[
V(x)
+
1-G(x)
\right],      
\end{align}
where we are using the notation
\begin{align*}
    U(x)\equiv        -{g'(x)\over g^2(x)}.
\end{align*}

{\color{black} Observe that, for every $\beta>0$, the function $V(x)$ is decreasing in $x$.}
Let $A$ be the set of distributions $G$ satisfying the condition: $U(x)$ is increasing in $x$.

If $G\in A$, then the equation \eqref{crit-point} has a unique solution, denoted by $x_0$. In other words, $x_0$ is the unique critical point for the function $f$. But since $\lim_{x\to 0^+}f(x)=\lim_{x\to \infty}f(x)=0$, it follows that $x_0$ is the unique mode of the density $f$. That is, under the above assumptions, the DT-G model is unimodal.
{\color{black}
\begin{remark}
Note that the set $A$ is nonempty. In fact, Table~\ref{tab:examplesA} presents several distributions that belong to this set.
\begin{table}[h!]
\centering
\caption{Examples of positive-support distributions such that 
$U(x)=-\frac{g'(x)}{g^2(x)}$ is increasing.}
\begin{tabular}{llll}
\toprule
Distribution & Density $g(x)$ & Condition & Remark \\
\midrule
Exponential &
$\lambda e^{-\lambda x}$ &
$\lambda>0$ &
$U(x)=e^{\lambda x}$ increasing \\[6pt]

Gamma &
$\dfrac{\beta^\alpha}{\Gamma(\alpha)}
x^{\alpha-1}e^{-\beta x}$ &
$\alpha\geqslant 1$, $\beta>0$ &
Includes exponential and Erlang \\[10pt]

Weibull &
$k\lambda (\lambda x)^{k-1}e^{-(\lambda x)^k}$ &
$k\geqslant 1$, $\lambda>0$ &
Case $k=2$: Rayleigh \\[6pt]

Erlang &
Gamma with $\alpha\in\mathbb{N}$ &
$\alpha\geqslant 1$ &
Classical reliability model \\

\bottomrule
\end{tabular}
\label{tab:examplesA}
\end{table}
\end{remark}
}

\subsection{Stochastic representation} \label{Stochastic representation}

It is simple to observe that $F$ in \eqref{cdf-main} can be written as (for $x>0$)
\begin{align*}
    F(x)
    =
    \int_0^{G(x)} f_Y(y){\rm d}y
    =
    \mathbb{P}(Y\leqslant G(x)),
\end{align*}
where $f_Y(y)=1-F_W(1-y)+yf_W(1-y)$, $0<y<1$, is the PDF of a random variable $Y$, with $W\sim{\rm Weibull}(1,\beta)$, that is, $F_W(w)=1-e^{-w^\beta}$ and $f_W(w)=F_W'(w)$, $w>0$. 
Hence, if $X$ has CDF $F$, then 
\begin{align*}
    F_X(x)=F(x)=\mathbb{P}(G^{-1}(Y)\leqslant x), \quad \forall x>0.
\end{align*}
In other words,
\begin{align}\label{rep-stochastic}
    X\stackrel{d}{=}G^{-1}(Y),
\end{align}
where $\stackrel{d}{=}$ means equality in distribution of two random variables.

{\color{black}
\begin{remark}
Note that the stochastic representation in \eqref{rep-stochastic} is not unique. 
For example, let $U\sim U(0,1)$ and define $V\in(0,1)$ as the solution of
\[
U = V\exp\!\left(-(1-V)^{\beta}\right).
\]
Then $X\stackrel{d}{=}G^{-1}(V)$. 
This stochastic representation is useful for random number generation and has the advantage over \eqref{rep-stochastic} of avoiding the introduction of an auxiliary density.
\end{remark}
}

\subsection{Identifiability}

Let $\boldsymbol{\varrho}$ be the parameter vector corresponding to the CDF $G$ in \eqref{cdf-main}. 
In this section we will denote $G(x;\boldsymbol{\varrho})$ instead of $G(x)$, and $F(x;\beta,\boldsymbol{\varrho})$ instead of $F(x)$.

Suppose {\color{black} that the baseline family $G(x;\boldsymbol{\varrho})$ is identifiable in 
$\boldsymbol{\varrho}$.}

In what follows we will show that the function $(\beta,\boldsymbol{\varrho})\mapsto F(x;\beta,\boldsymbol{\varrho})$, $\forall x>0$, is injective.

Indeed, suppose that
\begin{align*}
  F(x;\beta_1,\boldsymbol{\varrho}_1)
  =
  F(x;\beta_2,\boldsymbol{\varrho}_2), \quad \forall x>0.
\end{align*}
Equivalently,
\begin{align}\label{assumption-1}
    G(x;\boldsymbol{\varrho}_1)\,
    e^{-\bar{G}(x;\boldsymbol{\varrho}_1)^{\beta_1}}
    =
    G(x;\boldsymbol{\varrho}_2)\,
    e^{-\bar{G}(x;\boldsymbol{\varrho}_2)^{\beta_2}},
    \quad \forall x>0 .
\end{align}

{\color{black}
Let $y_i(x)=G(x;\boldsymbol{\varrho}_i)\in(0,1)$, $i=1,2$. Then \eqref{assumption-1} can be written as
\begin{equation}\label{eq-id-1}
y_1(x)\,e^{-(1-y_1(x))^{\beta_1}}
=
y_2(x)\,e^{-(1-y_2(x))^{\beta_2}},
\quad \forall x>0 .
\end{equation}

Consider the map
\[
\phi(y,\beta)=y\,e^{-(1-y)^{\beta}},\quad 0<y<1,\ \beta>0 .
\]
Its partial derivatives are
\[
\frac{\partial\phi}{\partial y}
=
e^{-(1-y)^{\beta}}
\Bigl[1+\beta y(1-y)^{\beta-1}\Bigr]>0,
\quad
\frac{\partial\phi}{\partial\beta}
=
-y(1-y)^{\beta}\ln(1-y)\,
e^{-(1-y)^{\beta}}>0 ,
\]
since $\log(1-y)<0$ for $0<y<1$. Hence, $\phi(y,\beta)$ is strictly increasing in $y$ for fixed $\beta$, and strictly increasing in $\beta$ for fixed $y$.

Since \eqref{eq-id-1} holds for all $x$ and $\phi(y,\beta)$ is strictly increasing in $y$ for fixed $\beta$, any difference between 
$G(x;\boldsymbol{\varrho}_1)$ and $G(x;\boldsymbol{\varrho}_2)$ would produce a strict inequality. Therefore,
\[
G(x;\boldsymbol{\varrho}_1)
=
G(x;\boldsymbol{\varrho}_2),
\quad \forall x>0 .
\]
By identifiability of the baseline family $G$, it follows that
$
\boldsymbol{\varrho}_1=\boldsymbol{\varrho}_2.
$

With $y(x)=G(x;\boldsymbol{\varrho}_1)$ fixed, strict monotonicity of $\phi$ in $\beta$ then yields
$
\beta_1=\beta_2.
$

Hence $(\boldsymbol{\varrho}_1,\beta_1)=(\boldsymbol{\varrho}_2,\beta_2)$, proving identifiability.
}

\subsection{Quantiles}

The quantile function of $X$ with CDF $F(x)$, denoted $Q_X(p)$, $0 < p < 1$, facilitates random number generation and is invariant under monotonic transformations. Consequently, from stochastic representation \eqref{rep-stochastic}, we have
\begin{align*}
    Q_X(p)=G^{-1}(Q_Y(p)),
\end{align*}
where $Y$ has PDF $f_Y(y)$ given in Subsection \ref{Stochastic representation}. Note that $Q_Y(p)$ can be obtained by inverting the equation  
\begin{eqnarray}
 Q_Y(p) e^{-[1-Q_Y(p)]^{\beta}}=p.   
\label{quantileF}\end{eqnarray}

\subsection{Moments and truncated moments of DT-G random variables}

{\color{black}
Depending on the choice of the baseline CDF $G$, this subsection derives explicit formulas for the moments and truncated moments of a DT--$G$ random variable. 
To this end, we state a general result whose proof follows from Lemma~1 of \cite{math11092172}; for completeness and the reader’s convenience, we provide the details.
}
\begin{proposition}\label{Moments-truncated}
   Let $X$ be a real-valued random variable 
   and $p>0$ be a real number. For all $\varepsilon\geqslant 0$ and $\delta>0$, the following identity is satisfied:
   \begin{align*}
   \mathbb{E}(X^p\mathds{1}_{\{\varepsilon<X<\delta\}})
	=
	\varepsilon^{p}
	\mathbb{P}(X>\varepsilon)
	-
	\delta^{p}
	\mathbb{P}(X>\delta)
	+
	p
	\int_{\varepsilon}^{\delta}
	u^{p-1}
	\mathbb{P}(X>u)
	{\rm d}u.
\end{align*}
\end{proposition}
\begin{proof}
    For all $\varepsilon\geqslant 0$ and $\delta>0$, note that the random variable $X^p\mathds{1}_{\{\varepsilon<X<\delta\}}$ is non-negative and integrable. Then, by using the well-known formula $\mathbb{E}(Z)=p\int_0^\infty z^{p-1}\mathbb{P}(Z>z){\rm d}z$, with $\mathbb{P}(Z>0)=1$ and $p>0$, we have
\begin{align}\label{id-1-1}
	\mathbb{E}(X^p\mathds{1}_{\{\varepsilon<X<\delta\}})
	&=
	p
	\int_{0}^{\infty}
	u^{p-1}
	\mathbb{P}(X\mathds{1}_{\{\varepsilon<X<\delta\}}>u)
	{\rm d}u \nonumber
	\\[0,2cm]
	&=
	p
	\int_{0}^{\varepsilon}
	u^{p-1}
	\mathbb{P}(X\mathds{1}_{\{\varepsilon<X<\delta\}}>u)
	{\rm d}u
	+
	p
	\int_{\varepsilon}^{\infty}
	u^{p-1}
	\mathbb{P}(X\mathds{1}_{\{\varepsilon<X<\delta\}}>u)
	{\rm d}u.
\end{align}

Note that
\begin{align*}
	X\mathds{1}_{\{\varepsilon<X<\delta\}}>u 
	\quad \Longleftrightarrow \quad
	X>u \ \text{and} \ \varepsilon<X<\delta
	\quad \Longleftrightarrow \quad
	\max\{u,\varepsilon\}<X<\delta.
\end{align*}
Hence, \eqref{id-1-1} can be written as
\begin{align*}
	&=
	p
	\int_{0}^{\varepsilon}
	u^{p-1}
	\mathbb{P}(\max\{u,\varepsilon\}<X<\delta)
	{\rm d}u
	+
	p
	\int_{\varepsilon}^{\delta}
	u^{p-1}
	\mathbb{P}(\max\{u,\varepsilon\}<X<\delta)
	{\rm d}u
	\\[0,2cm]
	&=
	p
	\int_{0}^{\varepsilon}
	u^{p-1}
	\mathbb{P}(\varepsilon<X<\delta)
	{\rm d}u
	+
	p
	\int_{\varepsilon}^{\delta}
	u^{p-1}
	\mathbb{P}(u<X<\delta)
	{\rm d}u
	\\[0,2cm]
	&=
\varepsilon^{p}
\mathbb{P}(X>\varepsilon)
-
\delta^{p}
\mathbb{P}(X>\delta)
+
p
\int_{\varepsilon}^{\delta}
u^{p-1}
\mathbb{P}(X>u)
{\rm d}u.
\end{align*}
This completes the proof.
\end{proof}

Now, if $X$ has a DT-G distribution, by \eqref{cdf-main} and by Proposition \ref{Moments-truncated}, we have (for all $\varepsilon\geqslant 0$ and $\delta>0$)
   \begin{align}\label{int-1}
   \mathbb{E}(X^p\mathds{1}_{\{\varepsilon<X<\delta\}})
	=
	\varepsilon^{p}
        \left[1-G(\varepsilon) e^{-\bar{G}(\varepsilon)^{\beta}}\right]
	-
	\delta^{p}
	\left[1-G(\delta) e^{-\bar{G}(\delta)^{\beta}}\right]
	+
	p
	\int_{\varepsilon}^{\delta}
	u^{p-1}
	\left[1-G(u) e^{-\bar{G}(u)^{\beta}}\right]
	{\rm d}u.
\end{align}
Assuming that $\mathbb{E}(\vert X\vert^p )<\infty$ and 
\begin{itemize}
    \item 
Letting $\delta\to\infty$, we get
   \begin{align}\label{int-2}
   \mathbb{E}(X^p\mathds{1}_{\{X>\varepsilon\}})
	=
	\varepsilon^{p}
        \left[1-G(\varepsilon) e^{-\bar{G}(\varepsilon)^{\beta}}\right]
	+
	p
	\int_{\varepsilon}^{\infty}
	u^{p-1}
	\left[1-G(u) e^{-\bar{G}(u)^{\beta}}\right]
	{\rm d}u.
\end{align}
   \item 
Letting $\delta\to\infty$ and  $\varepsilon\to 0^+$, we get
   \begin{align}\label{int-3}
   \mathbb{E}(X^p)
	=
	p
	\int_{0}^{\infty}
	u^{p-1}
	\left[1-G(u) e^{-\bar{G}(u)^{\beta}}\right]
	{\rm d}u.
\end{align}
\end{itemize}

Then, from \eqref{int-1}, \eqref{int-2} and \eqref{int-3}, in the special case $G(x)=1-e^{-\lambda x}, \, x>0, \, \lambda>0$ (exponential distribution), we obtain
\begin{enumerate}
    \item[a.] 
       \begin{align*}
   \mathbb{E}(X^p\mathds{1}_{\{\varepsilon<X<\delta\}})
	&=
	\varepsilon^{p}
    \left[1-(1-e^{-\lambda \varepsilon}) e^{-e^{-\lambda\beta \varepsilon}}\right]
    -
	\delta^{p}
	\left[1-(1-e^{-\lambda \delta}) e^{-e^{-\lambda\beta \delta}}\right]
    \\[0,2cm]
	&+
	p
	\int_{\varepsilon}^{\delta}
	u^{p-1}
	\left[1-(1-e^{-\lambda u}) e^{-e^{-\lambda\beta u}}\right]
	{\rm d}u.
\end{align*}
    \item[b.] 
       \begin{align*}
   \mathbb{E}(X^p\mathds{1}_{\{X>\varepsilon\}})
	=
	\varepsilon^{p}
   \left[1-(1-e^{-\lambda \varepsilon}) e^{-e^{-\lambda\beta \varepsilon}}\right]
   +
	p
	\int_{\varepsilon}^{\infty}
	u^{p-1}
	\left[1-(1-e^{-\lambda u}) e^{-e^{-\lambda\beta u}}\right]
	{\rm d}u.
\end{align*}
    \item[c.] 
    $$
       \mathbb{E}(X^p)
	=
	p
	\int_{0}^{\infty}
	u^{p-1}
	\left[1-(1-e^{-\lambda u}) e^{-e^{-\lambda\beta u}}\right]
	{\rm d}u.
    $$
\end{enumerate}
    The above integrals in a, b and c cannot be expressed in terms of standard mathematical functions. Therefore, for its respective calculation, numerical analysis is required.

\section{ {\color{black} DT regression models}}\label{DTReg}

{\color{black}
First, this section introduces two specific models from the DT family of distributions. Subsequently, we present two new regression models based on reparameterizations of these particular cases within the DT family. Parameter inference is conducted using maximum likelihood estimation, and goodness of fit is assessed through quantile residuals.
}

\subsection{DT-exponential distribution} \label{DTE} 

\noindent Replacing $G(x)=1-e^{-\lambda x}$ in (\ref{cdf-main}), we have obtained a specific sub-model which is called Dutta transformed exponential distribution (DTED) with two parameters $(\beta,\lambda)$. Then the CDF and PDF of DTED are given by 
\begin{align} \label{4.1}
    F_{DTED}(x)= (1-e^{-\lambda x}) e^{-e^{-\beta\lambda x}},~\mbox{where}~x>0, \beta,\lambda>0;
\end{align}
and 
\begin{align*}
    f_{DTED}(x)=  \lambda e^{-\lambda x} \left[1+\beta (1-e^{-\lambda x})e^{-(\beta-1)\lambda x}\right]e^{-e^{-\beta\lambda x}}, 
\end{align*}
respectively. 


\subsection{ {\color{black}DT-Weibull distribution}} \label{DTW}

{\color{black}
\noindent Replacing $G(x)= 1-e^{-(\lambda x)^k}$ in (\ref{cdf-main}), we have obtained a specific sub-model which is called Dutta transformed Weibull distribution (DTWD) with three parameters $(\beta,\lambda, k)$. The CDF and PDF of DTWD are given by 
\begin{eqnarray}\label{CDFDTW}
 F_{DTWD}(x)= [1 - e^{-(\lambda x)^k}]\exp\left\{-e^{-\beta(\lambda x)^k} \right\},~\mbox{where}~x>0, \beta,\lambda, k>0;
\end{eqnarray}
and
\begin{align*}
f_{DTWD}(x) &= k \lambda^k x^{k-1} e^{-(\lambda x)^k} \exp\Big\{ - e^{-\beta(\lambda x)^k}\Big\} 
 \left[ 1 + \beta \, (1 - e^{-(\lambda x)^k}) \, e^{-(\beta-1)(\lambda x)^k}\right],
\end{align*}
respectively.

}

\subsection{Reparametrizations}

{\color{black}
In the formulation of the DT-exponential and DT-Weibull distributions in the previous subsections, the original parameters do not correspond to direct, interpretable characteristics of the distribution. Specifically, if regression structures are assigned to \(\lambda\), $\beta$ and $k$, it becomes challenging to interpret the effects of covariates on the response because these parameters do not correspond to straightforward features such as the mean, mode, median, or any quantile. Consequently, the direction or magnitude of the effect (positive or negative) cannot be easily determined, limiting the practical interpretability of the regression model. To address this issue, we introduce median-based reparameterizations of the DT-exponential and DT-Weibull distributions.  This approach enables us to measure the dependence of the median response on covariates through the regression coefficients.
}

{\color{black} In the following, we present the construction of the reparameterization for the DTED model.}
Let $q_\tau = Q_{DTED}(\tau)$, $0<\tau<1$ represent the $\tau$-th quantile of DT-exponential distribution. When $\tau = 0.5$, $\mu = q_{0.5}$ corresponds to the median of the distribution.
From (\ref{4.1}), we have that the quantile $q_\tau$ satifies:
\begin{eqnarray}\label{tau}
	\tau = (1-e^{-\lambda q_\tau}) e^{-e^{-\beta\lambda q_\tau}}.
\end{eqnarray}

{\color{black}
To propose a new parameterization, we first isolate the value of \(\beta\) in Equation (\ref{tau}). To this end, we apply the natural logarithm to both sides of the equation and then use logarithmic properties, yielding
\begin{eqnarray*}
    \log\left(\frac{1-e^{-\lambda q_\tau}}{\tau}\right) = e^{-\beta\lambda q_\tau}.
\end{eqnarray*}

To further isolate the value of $\beta$, we apply the natural logarithm once again, obtaining
\begin{eqnarray}\label{betaf}
     \beta = -\frac{1}{\lambda}\log\left[\left(\log\left(\frac{1 - e^{-\lambda q_\tau}}{\tau}\right)\right)^{1/ q_\tau}\right].
\end{eqnarray}

The initial idea is to replace the parameter \(\beta\) in Equation (\ref{4.1}) with the expression in (\ref{betaf}). This substitution allows the CDF (and PDF) to depend on the \(\tau\)-th quantile, making at least one parameter a direct characteristic of the distribution. However, the expression in (\ref{betaf}) is well defined if $0 < \tau<1-e^{-\lambda q_\tau}$, which is equivalent to $\lambda > -\log(1-\tau)/q_\tau$. In particular, if $\tau=0.5$, it is necessary that $\lambda > \log(2)/q_{0.5}$. Hence, to establish a well-defined median parameterization for the DTED model, we must take this condition into account.

To achieve a median regression structure, we will proceed with the following parameterization of the DT-exponential distribution:
}
\begin{eqnarray*}
    \lambda = \frac{\sigma + \log(2)}{\mu}, \quad \beta = -\frac{1}{\sigma + \log(2)}\log\left(\log\left(\frac{1 - e^{-(\sigma + \log(2))}}{0.5}\right)\right),
\end{eqnarray*}
where $\mu>0$ denotes the median of DT-exponential distribution  {\color{black} and $\sigma>0$. Note that $\beta$ has the form given in (\ref{betaf}) and depends only on \(\sigma\), while \(\lambda\) is expressed as a function of $\mu$ and $\sigma$, incorporating the additional term $\log(2)$ to ensure that the parameterization is always valid.
}

The CDF e PDF can be written, in
the new parameterization, as
\begin{eqnarray} \label{RCDF}
   {\color{black} F_{RDTED}(x; \mu, \sigma)}= \left(1-\exp\left\{-\frac{(\sigma+\log(2))x}{\mu}\right\}\right)\exp\left\{- {\color{black}A(\sigma)}^{x/\mu} \right\},
\end{eqnarray}
{where} $x>0$, $\mu>0$, $\sigma>0$, with
{\color{black}
\begin{eqnarray*}
    A(\sigma) = \log\left(\frac{1-e^{-(\sigma+\log(2))}}{0.5}\right) = \log(2-e^{-\sigma}),
\end{eqnarray*}
}and 
\begin{align*}
{\color{black}f_{RDTED}(x; \mu, \sigma)} &= \frac{(\sigma+\log(2))}{\mu} \exp\left\{-\frac{(\sigma+\log(2))x}{\mu}\right\}  \Bigg[ 1 - \frac{{\color{black} A(\sigma)}^{{x}/{\mu}}}{\sigma+\log(2)} \log( {\color{black}A(\sigma)})\\ 
&\quad \times \left( 1 - \exp\left\{-\frac{(\sigma+\log(2))x}{\mu}\right\} \right)  \exp\left\{\frac{(\sigma+\log(2))x}{\mu}\right\} \Bigg] \exp\left\{ -  {\color{black} A(\sigma)}^{{x}/{\mu}} \right\}.
\end{align*}
{\color{black} This distribution is called the reparameterized DT-exponential distribution (RDTED) and is denoted by \( X \sim \mathrm{RDTED}(\mu, \sigma) \).
}
This approach will result in a simpler and more interpretable regression model based on the  DT-exponential distribution, with regressions structures assigned to the parameters $\mu$ and $\sigma$. {\color{black} Although $\mu$ denotes the median of the RTDED distribution, it also acts as a scale parameter, since \(F_{RDTED}(x; \mu, \sigma) = F_{RDTED}(x/\mu; 1, \sigma)\). In contrast, the parameter \(\sigma\) is neither a scale nor a location parameter; rather, it controls the shape of the distribution. }

{\color{black}
Figure \ref{densities} shows the PDF curves of the RDTED distribution, obtained by varying one parameter while keeping the other fixed. As expected, the parameter \(\mu\) affects the scale of the distribution. Larger values of \(\mu\) increase the scale, making the distribution more platykurtic. Additionally, as \(\mu\) increases, the distribution shifts slightly to the right, which is consistent with \(\mu\) representing the median of the distribution. On the other hand, it is clear that the parameter \(\sigma\) alters the shape of the distribution, producing either a unimodal or an inverted ‘J’ shape.

}
\begin{figure}[!htb]
	\centering
	\subfigure[Varyng $\mu$ with $\sigma=0.1$]{\includegraphics[width=8cm,height=7cm]{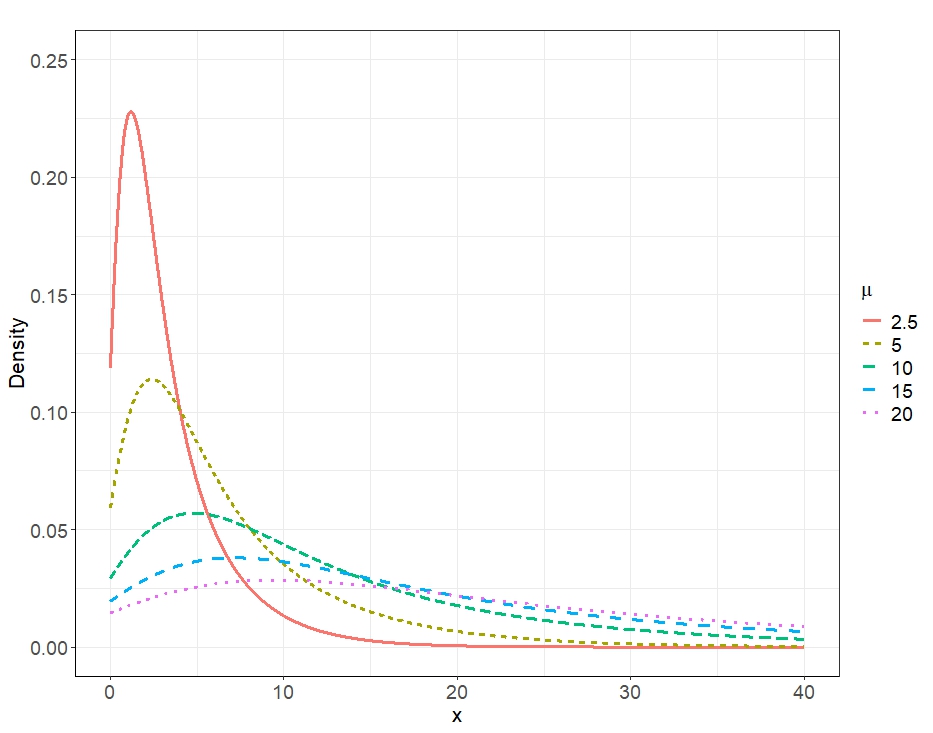}}
	\qquad
    \subfigure[Varyng $\sigma$ with $\mu=10$]
    {\includegraphics[width=8cm,height=7cm]{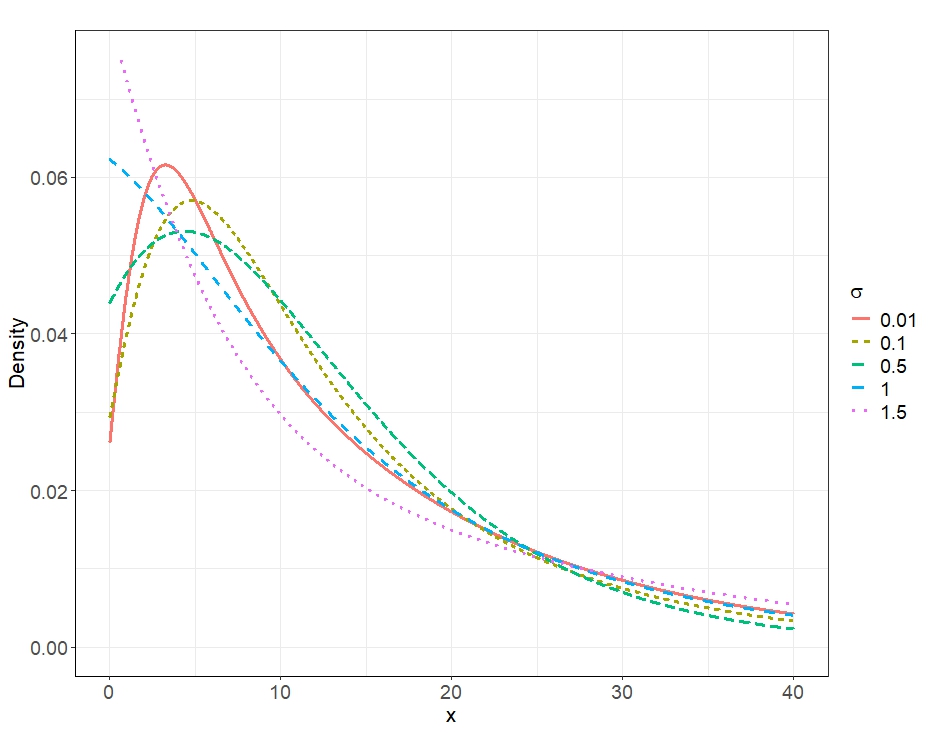}}
    \caption{{\color{black}RDTED densities for different values of $\mu$ and $\sigma$.} }
    \label{densities}
\end{figure}

{\color{black}
For the DTWD model, we propose the same type of reparameterization as that used for the DTED model. Consider $q_\tau = Q_{DTWD}(\tau)$, $0<\tau<1$ represent the $\tau$-th quantile of DT-Weibull distribution. 
From (\ref{CDFDTW}), we have that the quantile $q_\tau$ satifies:
\begin{eqnarray}\label{tau}
	\tau = [1 - e^{-(\lambda q_\tau)^k}]\exp\left\{-[e^{-(\lambda q_\tau)^k}]^\beta \right\}.
\end{eqnarray}
Using a similar strategy as before, we isolate the parameter \(\beta\) as follows:
\begin{eqnarray}\label{betaf2}
     \beta = -\frac{1}{\lambda ^k}\log\left[\left(\log\left(\frac{1 - e^{-(\lambda q_\tau)^k}}{\tau}\right)\right)^{1/ {(q_\tau}^{k})}\right],
\end{eqnarray}
provided that $0 < \tau<1-e^{-(\lambda q_\tau)^k}$ which is equivalent to $\lambda > [-\log(1-\tau)]^{1/k}/q_\tau$. In particular, if $\tau=0.5$, it is necessary that $\lambda > \log(2)^{1/k}/q_{0.5}$.

To achieve a well-defined median regression structure, we consider the following parameterization of the DT-Weibull distribution: 
\begin{eqnarray*}
 k=\nu, \quad     \lambda = \frac{\sigma + \log(2)^{1/\nu}}{\mu}, \quad \beta = - \frac{1}{\big(\sigma + (\log (2))^{1/\nu}\big)^\nu} 
\; \log \Bigg[ \log \Bigg( \frac{1 - \exp\Big\{- \big(\sigma + (\log( 2))^{1/\nu}\big)^\nu \Big\}}{0.5} \Bigg) \Bigg],
\end{eqnarray*}
where $\mu>0$ denotes the median of DT-Weibull distribution,  $\sigma>0$, and $\nu>0$. 
Note that $\beta$ has the form given in (\ref{betaf2}) and depends of \(\sigma\) and $\nu$, while \(\lambda\) is expressed as a function of $\mu$, $\sigma$, and $\nu$ incorporating the additional term $\log(2)^{1/\nu}$. If $\nu=1$, we obtain the RTDED model. Therefore, the RDTWD model generalizes the RTDED model.

The CDF e PDF can be written, in
the new parameterization, as
\begin{eqnarray} \label{RCDF2}
   {\color{black} F_{RDTWD}(x; \mu, \sigma, \nu)}= \left(1-\exp\left\{-\left(\frac{(\sigma+(\log(2))^{1/\nu})\,x}{\mu}\right)^\nu\right\}\right)
\exp\left\{-A(\sigma, \nu)^{(x/\mu)^\nu}\right\},
\end{eqnarray}
{where} $x>0$, $\mu>0$, $\sigma>0$, $\nu>0$, with
{\color{black}
\begin{eqnarray*}
    A(\sigma, \nu) = \log \Bigg( \frac{1 - \exp\Big\{- \big(\sigma + (\log(2))^{1/\nu}\big)^\nu \Big\}}{0.5} \Bigg),
\end{eqnarray*}
}and 
\begin{align*}
{\color{black}f_{RDTWD}(x; \mu, \sigma, \nu)} &= \nu \left(\frac{\sigma+(\log(2))^{1/\nu}}{\mu}\right)^\nu x^{\nu-1}
\exp\left\{-\left(\frac{(\sigma+(\log(2))^{1/\nu})\,x}{\mu}\right)^\nu\right\} \\
&\quad \times \Bigg[
1 - \frac{\log(A(\sigma, \nu))}{(\sigma+(\log(2))^{1/\nu})^\nu}
\left(1-\exp\left\{-\left(\frac{(\sigma+(\log(2))^{1/\nu})\,x}{\mu}\right)^\nu\right\}\right) \\
&\quad\times A(\sigma, \nu)^{(x/\mu)^\nu}
\exp\left\{\left(\frac{(\sigma+(\log(2))^{1/\nu})\,x}{\mu}\right)^\nu\right\}
\Bigg]\exp\left\{-A(\sigma, \nu)^{(x/\mu)^\nu}\right\} .
\end{align*}
This distribution is called the reparameterized DT-Weibull distribution (RDTWD) and is denoted by \( X \sim \mathrm{RDTWD}(\mu, \sigma, \nu) \).
For this distribution, the parameter \(\mu\) is a scale parameter because $ F_{RDTWD}(x; \mu, \sigma, \nu) = F_{RDTWD}(x/\mu; 1, \sigma, \nu)$, whereas the parameters \(\sigma\) and \(\nu\) are shape parameters.

}

{\color{black}
Figure \ref{densities_DTWD} displays the PDF curves of the RDTWD distribution, obtained by varying one parameter at a time while keeping the others fixed. As in the RDTED case, the parameter $\mu$ primarily controls the scale of the distribution, with larger values leading to an increased spread. In addition, higher values of $\mu$ induce a slight rightward shift in the distribution, reflecting changes in its central tendency (median). The parameters $\sigma$ and $\nu$, on the other hand, play a key role in shaping the distribution, allowing for a variety of forms, including inverted “J”-shaped, unimodal, and bimodal patterns.

}
\begin{figure}[!htb]
	\centering
	\subfigure[Varyng $\mu$ with $\sigma=0.1$ and $\nu=1.5$]{\includegraphics[width=8cm,height=7cm]{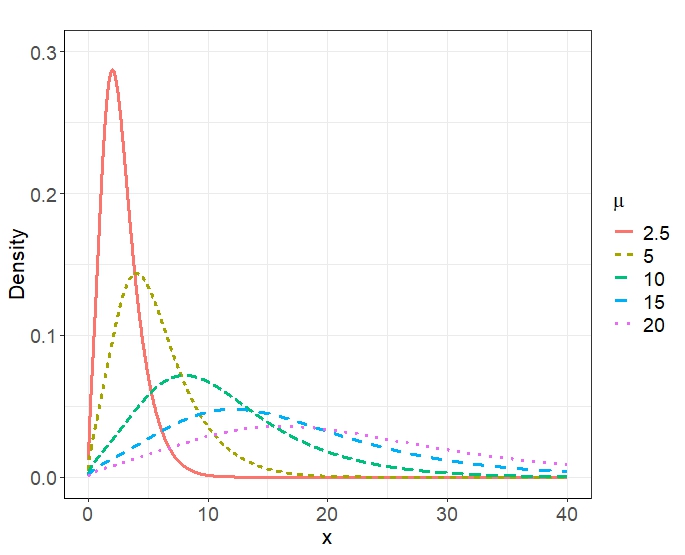}}
	\qquad
    \subfigure[Varyng $\sigma$ with $\mu=10$ and $\nu=2$]
    {\includegraphics[width=8cm,height=7cm]{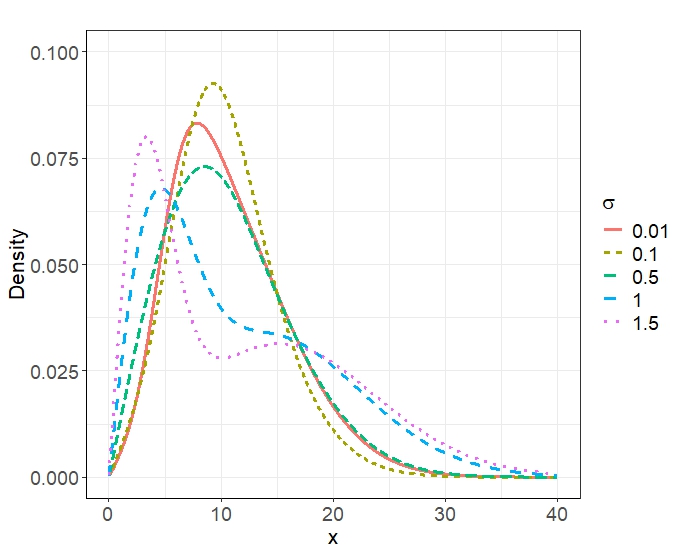}}\\
    \subfigure[Varyng $\nu$ with $\mu=10$ and $\sigma=1.0$]{\includegraphics[width=8cm,height=7cm]{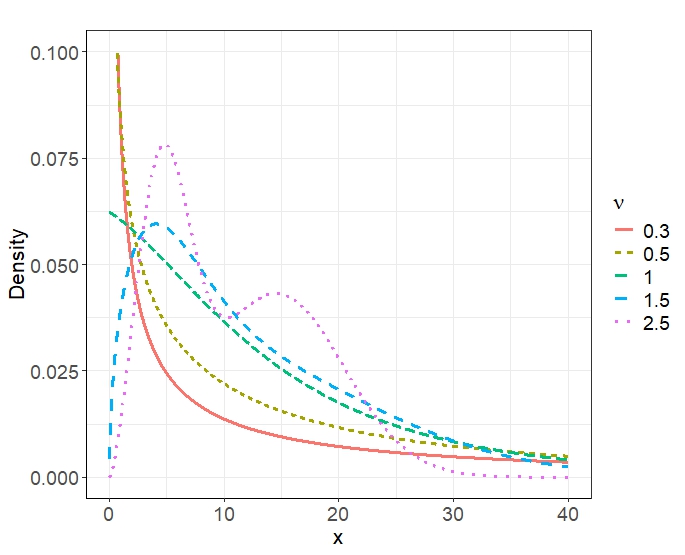}}
    \caption{{\color{black}RDTWD densities for different values of $\mu$,  $\sigma$ and $\nu$.} }
    \label{densities_DTWD}
\end{figure}

{\color{black}
One way to propose a reparameterization for other cases of the DT family is based on the general idea of isolating the value of \(\beta\) in the quantile function (\ref{quantileF}). By doing so, we express \(\beta\) as a function of the distribution's quantiles, allowing at least one parameter to have a simple interpretation, which is crucial in a regression context.
Next, the obtained expression for \(\beta\) is substituted into \(F(x)\) given in (\ref{cdf-main}), allowing the CDF to be expressed in terms of quantiles. In particular, for the DTED and DTWD models, we consider the median, i.e., the 50th percentile. However, certain precautions must be taken to ensure that the inversions used to derive the expression for \(\beta\) are well-defined. In the DTED and DTWD cases, it was also necessary to rewrite \(\lambda\), enforcing the condition to ensure that \(\beta\) is well-defined. In summary, reparameterizations similar to those proposed in this work could be applied to other cases within the DT family, ensuring that the resulting CDFs and PDFs are well-defined.
}

\subsection{Regression structure}

{\color{black}
To establish the regression structures under RTDED and RTDWD,} consider the following. Let $Y_1,\ldots,Y_n$ be $n$ independent random variables representing $n$ response variables, where {\color{black} $Y_i\sim {\rm RDTWD}(\mu_i, \sigma_i, \nu_i)$. If $\nu_i=1$, $\forall i=1,\ldots,n$, we have $Y_i\sim {\rm RDTED}(\mu_i, \sigma_i)$.}
The
model is obtained by assuming that $\mu_i$, $\sigma_i$ {\color{black} and $\nu_i$} can be written as
\begin{eqnarray*}
	g_\mu(\mu_i)= {\boldsymbol W}_i^\top{\boldsymbol \alpha} = \sum_{j=1}^{p_1} w_{ij}\alpha_j, \\ 
 g_\sigma(\sigma_i) = {\boldsymbol Z}_i^\top{\boldsymbol \gamma} =\sum_{j=1}^{p_2} z_{ij}\gamma_j,\\
{\color{black} g_\nu(\nu_i) = {\boldsymbol S}_i^\top{\boldsymbol \xi} =\sum_{j=1}^{p_3} s_{ij}\xi_j,}
\end{eqnarray*}
where ${\boldsymbol \alpha}= (\alpha_1,\ldots,\alpha_{p_1})^\top\in\mathbb{R}^{p_1}$, ${\boldsymbol \gamma}= (\gamma_1,\ldots,\gamma_{p_2})^\top\in\mathbb{R}^{p_2}$ and {\color{black} ${\boldsymbol \xi}= (\xi_1,\ldots,\xi_{p_3})^\top\in\mathbb{R}^{p_3}$ } are vectors of unknown regression parameters, ${\boldsymbol W}_i=(w_{i1},\ldots,w_{ip_1})^\top$, ${\boldsymbol Z}_i=(z_{i1},\ldots,z_{ip_2})^\top$ {\color{black} and ${\boldsymbol S}_i=(s_{i1},\ldots,s_{ip_3})^\top$} are vectors of covariates of lengths $p_1$, $p_2$ {\color{black} and $p_3$  ($p_1+p_2+p_3 =p<n$),} respectively, which
are assumed fixed and known. In addition, $g_\mu(\cdot): (0,\infty)\rightarrow \mathbb{R}$, $g_\sigma(\cdot): (0,\infty)\rightarrow \mathbb{R}$ {\color{black} and $g_\nu(\cdot): (0,\infty)\rightarrow \mathbb{R}$} are strictly monotonic and twice differentiable link functions.
There are several possible choices for the link functions. For instance,
one can use the identity, inverse, square root, and logarithm link function defined as $g(x)=x$, $g(x)=1/x$, $g(x)=\sqrt{x}$, and $g(x)=\log(x)$, respectively. Since that $\mu>0$, $\sigma>0$, {\color{black} and $\nu>0$}, the logarithm link is the most appropriate, i.e, $g_\mu(x) = g_\sigma(x) { \color{black} = g_\nu(x)}=\log(x)$. In particular, when using the logarithmic link function, each regression coefficient represents the log-relative change in the parameter ($\mu$, $\sigma$ {\color{black}or $\nu$}) for a one-unit increase in the covariate, holding other variables constant; that is, exponential of the coefficient gives the multiplicative effect on the parameter.

{\color{black}
Modeling the parameters \(\sigma\) and $\nu$ as functions of covariates in the RDTWD model enhances flexibility in fitting the data. However, the coefficients $\gamma_j$ and $\xi_j$ associated with these submodels have no practical interpretation, since \(\sigma\) and $\nu$ serve as shape parameters. Therefore, we recommend using the RDTWD (and RDTED) model with regression structures on all parameters when the goal is to improve the flexibility of the fit, rather than to interpret all parameters. The only meaningful interpretations pertain to the regression coefficients $\alpha_j$, $j=1,\ldots, p_1$, associated with the median response.

}

\subsection{Inference and diagnostics}

The vector of unknown parameters for the {\color{black}RDTWD} regression model is represented as {\color{black} ${\boldsymbol \theta} = ({\boldsymbol \alpha}^\top, {\boldsymbol \gamma}^\top, {\boldsymbol \xi}^\top)^\top \in\mathbb{R}^{p}$}. 
The inference procedure adopted here will be based on maximum likelihood estimation.
The log-likelihood function based on a sample of $n$ independent observations $y_1, \dots,  y_n$ is given by
\begin{eqnarray}
	\ell({\boldsymbol\theta}) = \sum_{i=1}^{n}\ell_i(\boldsymbol{\theta}),
\label{loglik}\end{eqnarray}
where {\color{black} $ \ell_i(\boldsymbol{\theta}) = \log[f_{RDTWD}(y_i; \mu_i, \sigma_i, \nu_i)]$ with $f_{RDTWD}(\cdot; {\color{black}\mu_i, \sigma_i, \nu_i})$ denoting the PDF of RDTWD evaluated in the parameters $\mu_i$, $\sigma_i$ and $\nu_i$.} The maximum likelihood estimator (MLE) is obtained by maximizing the log-likelihood function in (\ref{loglik}) with respect to $\boldsymbol{\theta}$. The MLE for $\boldsymbol{\theta}$ is denoted by $\widehat{\boldsymbol{\theta}}$.

The contribution of the $i$-th observation in the log-likelihood function is given by  
{\color{black}
\begin{eqnarray*}
    \ell_i(\boldsymbol{\theta}) &=& 
\log(\nu_i) +\nu_i\log\left(\frac{\sigma_i+(\log(2))^{1/\nu_i}}{\mu_i}\right) +(\nu_i-1) \log(y_i)
-\left(\frac{(\sigma_i+(\log(2))^{1/\nu_i})\,y_i}{\mu_i}\right)^{\nu_i} \\
\quad &+& \log\Bigg[
1 - \frac{\log(A(\sigma_i, \nu_i))}{(\sigma_i+(\log(2))^{1/\nu_i})^{\nu_i}}
\left(1-\exp\left\{-\left(\frac{(\sigma_i+(\log(2))^{1/\nu_i})\,y_i}{\mu_i}\right)^{\nu_i}\right\}\right) A(\sigma_i, \nu_i)^{(y_i/\mu_i)^{\nu_i}}\\
\quad &\times&\exp\left\{\left(\frac{(\sigma_i+(\log(2))^{1/\nu_i})\,y_i}{\mu_i}\right)^{\nu_i}\right\}
\Bigg] -A(\sigma_i, \nu_i)^{(y_i/\mu_i)^{\nu_i}}
\end{eqnarray*}
}

As solving the first-order conditions of the log-likelihood function yields a complex system of nonlinear equations, numerical methods are required to compute the maximum likelihood estimate of $\boldsymbol{\theta}$. To obtain the estimates, we use
 the {\tt gamlss} function in the package with the same name \citep{stasinopoulos2008generalized} contained in {\tt R} \citep{r2025}. Such framework considers maximum likelihood estimation procedure based on the Rigby and Stasinopoulos algorithm for the GAMLSS \citep{stasinopoulos2008generalized, stasinopoulos2017flexible}.
As the {\tt gamlss} function requires initial values to the parameters, {\color{black} for RDTED model} we consider $ \mu_i^{(0)} = {\rm med}(y)$ and $\sigma_i^{(0)} = 1$, $i = 1, \ldots, n$, with ${\rm med}(y)$ denoting the sample median. {\color{black} These starting values remain the same whether we consider modeling only the median or modeling both parameters through covariates.
For RDTWD models that incorporate covariates solely in the median submodel, we recommend using the maximum likelihood estimates obtained from RDTED regression models with the same specifications and \(\nu_i^{(0)} = 1\) for \(i = 1, \ldots, n\).
For RDTWD models that incorporate covariates in all submodels, we consider two approaches. First, we suggest using maximum likelihood estimates obtained from RDTWD regression models with a simpler specification; that is, covariates included only in \(\mu\) or in both \(\mu\) and \(\sigma\). The second approach involves using maximum likelihood estimates obtained from RDTED regression models with the same specifications in \(\mu\) and \(\sigma\), and setting \(\nu_i^{(0)} = 1\) for \(i = 1, \ldots, n\).
} 
 The {\tt gamlss} package offers a user-friendly interface for computing estimates and their standard errors, quantile residuals, predictions, and more. Additionally, the {\tt gamlss.ggplots} package \citep{gamlssggplot2} available in {\tt R} provides a collection of functions for visualizing GAMLSS objects with high-quality and customizable graphics.
Further details on how to implement a new probability distribution as a {\tt gamlss.family} object for use in GAMLSS modeling can be found in \cite{stasinopoulos2017flexible}.

To obtain standard errors of the maximum likelihood estimates and to perform hypothesis testing, we will rely on the asymptotic distribution of the maximum likelihood estimators. Under usual regularity conditions for maximum likelihood estimation, when the sample size $n$ is large, $\widehat{\boldsymbol{\theta}} \stackrel{ {\rm a} }{ \sim } {\rm N}_{p}(\boldsymbol{\theta}, J^{-1}(\boldsymbol{\theta}) ),$ where $\stackrel{ {\rm a} }{ \sim }$ denotes an approximate distribution, ${\rm N}_{p}(\cdot, \cdot)$ denotes the $p$-dimensional normal distribution, and $J^{-1}(\boldsymbol{\theta})$ is the inverse of the observed information matrix. Based on this, consider testing the hypotheses $H_0$: $\theta_j = \theta_j^0$ versus $H_1$: $\theta_j \neq \theta_j^0$, where $\theta_j^0$ is a specific value for the unknown parameter $\theta_j$, $j = 1, \ldots, p$. Under $H_0$, $z_j = (\widehat{\theta}_j - \theta_j^0)/{\rm se}(\widehat{\theta}_j)\stackrel{ {\rm a} }{ \sim } {\rm N}(0,1),$ where $\widehat{\theta}_j$ is the MLE of $\theta_j$, and ${\rm se}(\widehat{\theta}_j) = \sqrt{J_{jj}}$ is the asymptotic standard error of $\widehat{\theta}_j$, with $J_{jj}$ the $j$-th element of the diagonal of $J^{-1}(\boldsymbol{\theta})$. The test statistic $z_j$ is the square root of Wald's statistic \citep{wald1943}, which is widely used to perform hypothesis tests in regression models. In the next section, we will apply this test to evaluate the significance of the parameters in the {\color{black} RDTWD (and RDTED)} regression model, setting $\theta_j^0 = 0$.

To assess the goodness of fit of the {\color{black} RDTWD (and RDTED)} regression model, we will use the quantile residual introduced by \cite{dunn1996randomized}, defined as follows:
{\color{black}
\begin{eqnarray*}\label{residuals}
r_i =
    \Phi^{-1}(\widehat{F}_{RDTWD}(y_i; \widehat{\mu}_i, \widehat{\sigma}_i, \widehat{\nu}_i)),
\end{eqnarray*}
}where $\Phi^{-1}(\cdot)$ is the CDF of the standard normal distribution, and {\color{black} $\widehat{F}_{RDTWD}(y_i; \widehat{\mu}_i, \widehat{\sigma}_i, \widehat{\nu}_i)$} denotes the CDF in (\ref{RCDF2}) evaluated in $y_i$, $\widehat{\mu}_i$, $\widehat{\sigma}_i$, {\color{black}and $\widehat{\nu}_i$}. If the {\color{black} RDTWD} regression model is well-fitted, $r_i$ is approximately distributed as a standard normal distribution.
To evaluate the empirical distribution of quantile residuals within the GAMLSS framework, it is common practice to use worm plots \citep{buuren2001worm}. {\color{black} If the RDTWD regression model is well-fitted, the points in the worm plot should be close to the middle horizontal line, with 95\% of them lying between the upper and lower dotted curves indicating no systematic departure \citep{stasinopoulos2017flexible}. Additionally, Q-Q boxplots \citep{rodu2022q} can be used to assess the normality of residuals \citep{deCarvalho2026, otiniano2026regression}. This method incorporates tail information from a Q-Q plot into a traditional boxplot by replacing the boxplot’s whiskers with the Q-Q plot tails and displaying these tails alongside confidence bands.
These diagnostic plots can be created using the {\tt gamlss.ggplots} \citep{gamlssggplot2} and {\tt qqboxplot} \citep{qqbpackage} packages in {\tt R}.}

{\color{black} Furthermore, we will employ the Akaike Information Criterion (AIC) \citep{Akaike1974}, defined as $ {\rm AIC} =-2\ell(\widehat{\boldsymbol{\theta}}) + 2p$, and the Bayesian Information Criterion (BIC) \citep{Schwarz1978}, defined as $ {\rm BIC} =-2\ell(\widehat{\boldsymbol{\theta}}) + p\log(n)$, to compare alternative regression models. Lower AIC and BIC values indicate a better fit to the data.}

\section{Real data application}\label{App}

In this section, we provide a real data application related to 385 measurements from small-leaved lime trees (Tilia cordata) growing in Russia. The data set called {\tt lime} is available in the package {\tt GLMsData} \citep{GLMsDatapackage} of software {\tt R}. Our goal here is to model the foliage biomass measured in kg (oven dried matter) as a function of age of the tree in years, and the origin of the tree (coppice, natural, planted).

{\color{black}
Figure \ref{foliageplots} displays the plots of foliage biomass versus the covariates origin and age of the tree. These plots suggest that these variables may be relevant in explaining foliage biomass. In particular, older trees tend to have greater foliage biomass, whereas trees with planted origin tend to exhibit lower foliage biomass.
}
\begin{figure}[!htb]
	\centering
	\subfigure{\includegraphics[width=8cm,height=7cm]{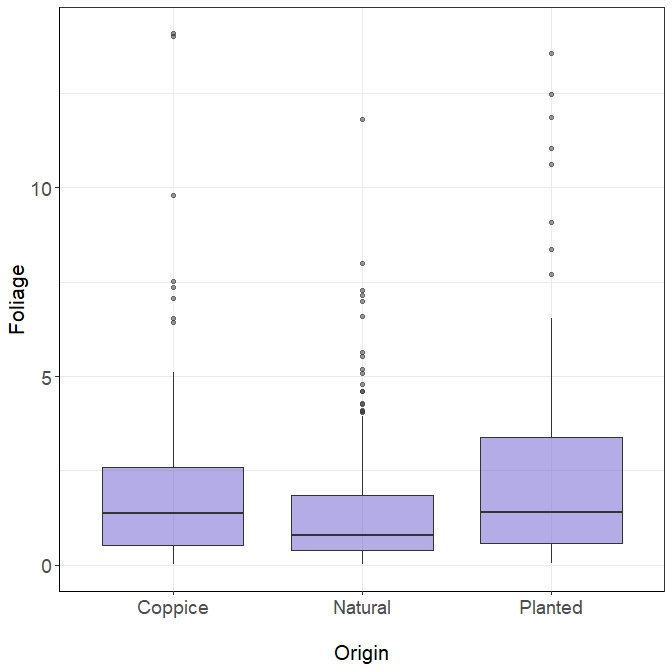}}
	\qquad
	\subfigure{\includegraphics[width=8cm,height=7cm]{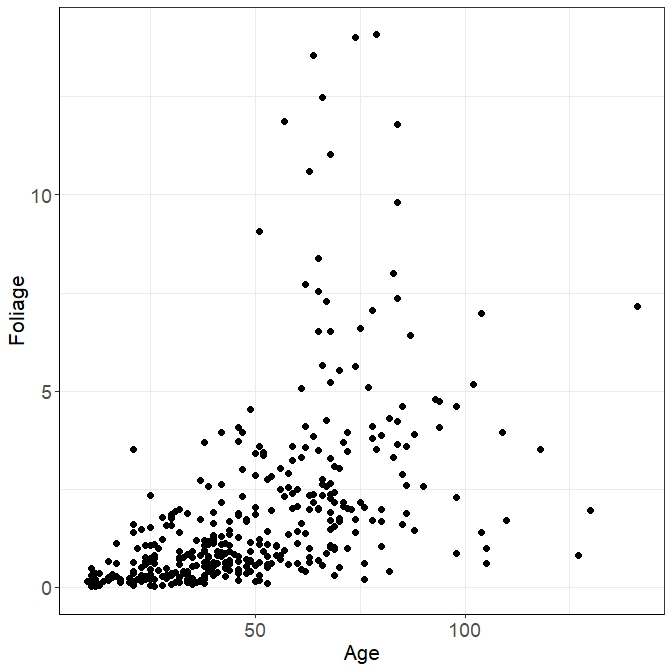}}
	\caption{ {\color{black} Plots of foliage biomass as a function of origin and age.} }
	\label{foliageplots}
\end{figure}

{\color{black} To demonstrate the advantages of the proposed RDTED and RDTWD models compared to alternative regression models commonly used, we also fitted gamma, inverse Gaussian, Weibull and log-normal regression models, all of which are implemented within the GAMLSS framework. Notably,  gamma and inverse Gaussian models are special cases of generalized linear models (GLMs) \citep{nelder1972generalized}, which are widely used to analyze positive, unimodal continuous data. In the gamma, inverse Gaussian, and Weibull models, the mean of the response is modeled. However, in the log-normal model, the median of the response is modeled, similar to the RDTED and RTDWD models. Except for the RDTWD model, which has three parameters, all the other models considered have two parameters.}

{\color{black}The regression models under consideration specify the parameter \(\mu_i\), which represents the $i$th response median or mean, as follows:}
\begin{eqnarray*}
\log(\mu_i) = \alpha_1 + \alpha_2{\rm age}_i + \alpha_3 {\rm originN}_i + \alpha_4{\rm originP}_i,
\end{eqnarray*}
where $i=1, \ldots, 385$, ${\rm age}_i$ is the age of tree $i$, ${\rm originN}_i=1$ if the origin of tree $i$ is natural, ${\rm originP}_i=1$ if the origin is planted, and ${\rm originN}_i={\rm originP}=0$ if the origin is coppice.
{\color{black} The second parameter, denoted by \(\sigma\) in all models, is initially assumed to be constant (i.e., not influenced by the covariates) and is expressed as $\log(\sigma)=\gamma_1$. In particular, for the RDTWD model is considered $\log(\nu)=\xi_1$.}

{\color{black}
Table \ref{tab:aicbic} shows the AIC and BIC values for the fitted regression models for the foliage data. According to the AIC results, the preferred regression model is the RDTED model followed by the RTDWD, gamma, Weibull, log-normal, and IG models. Based on the BIC results, the best model remains the same; the only difference is the ranking, where the positions of the gamma and RTDWD models are swapped.}

\begin{table}[!htb]
\centering
\caption{{\color{black} AIC and BIC values for the fitted regression models for foliage data.}}
\label{tab:aicbic}
\begin{tabular}{l|cc}
\toprule
\textbf{Model} & \textbf{AIC} & \textbf{BIC} \\\hline
RDTED                     & 995.86  &  1015.63\\
RDTWD                     & 996.32  &  1020.04\\
Gamma                     & 998.56  &  1018.33 \\
Inverse Gaussian          & 1288.19 &  1307.96 \\
Weibull                   & 1007.54 &  1027.31 \\
Log-normal                & 1016.67 &  1036.44 \\
\bottomrule
\end{tabular}
\end{table}
{\color{black} 
To assess the goodness of the fitted regression models, we consider worm plot and Q--Q boxplots of the quantile residuals, as shown in {\color{black}  Figures \ref{diagplots_M11} and \ref{diagplots_M12}}.
\begin{figure}[!htb]
	\centering
	\subfigure[Worm plot for RDTED model]{\includegraphics[width=6cm,height=5cm]{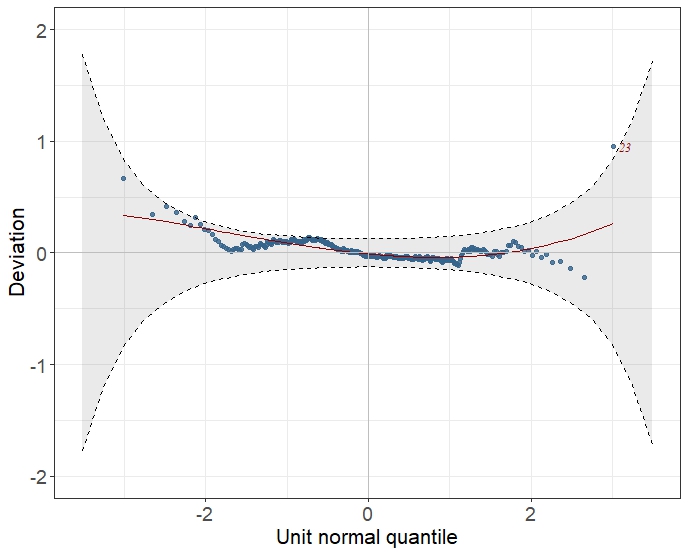}}
	\qquad
    \subfigure[Worm plot for RDTWD model]
    {\includegraphics[width=6cm,height=5cm]{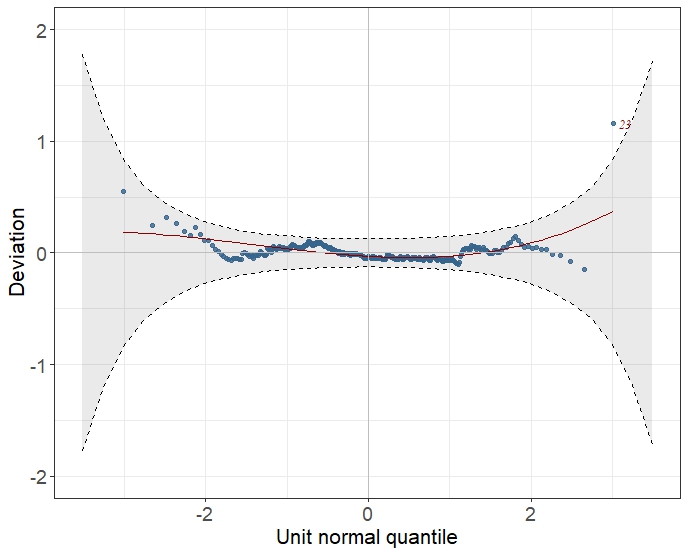}}\\
    \subfigure[Worm plot for Gamma model]
    {\includegraphics[width=6cm,height=5cm]{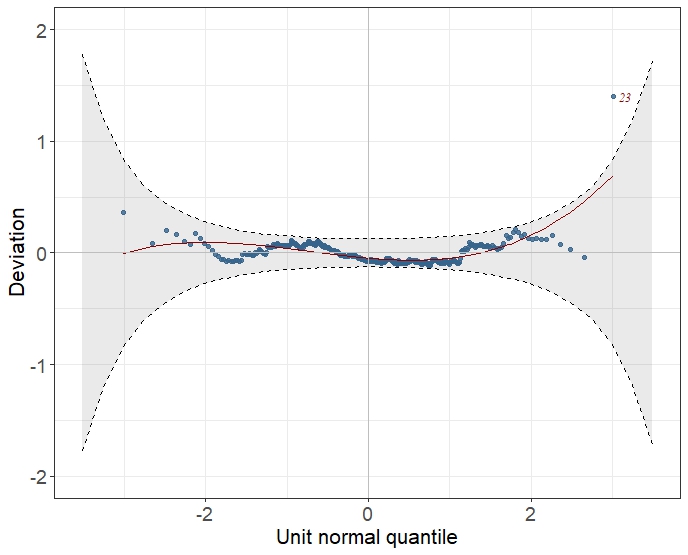}}
	\qquad
    \subfigure[Worm plot for IG model]
    {\includegraphics[width=6cm,height=5cm]{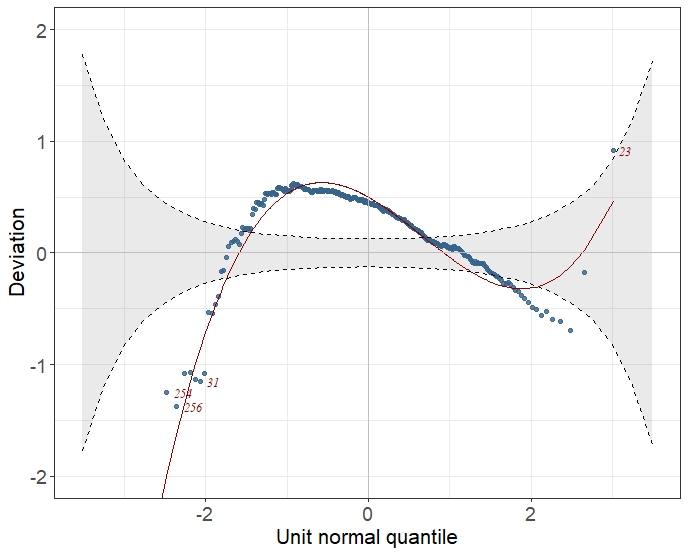}}\\
    \subfigure[Worm plot for Weibull model]
    {\includegraphics[width=6cm,height=5cm]{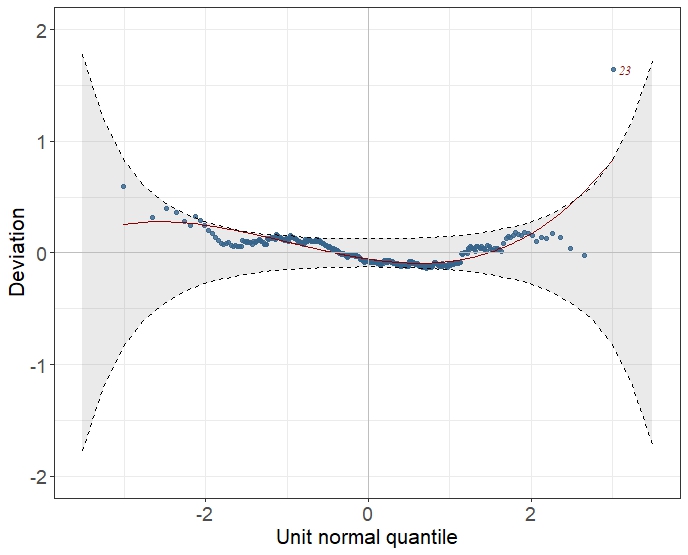}}
	\qquad
    \subfigure[Worm plot for Log-normal model]
    {\includegraphics[width=6cm,height=5cm]{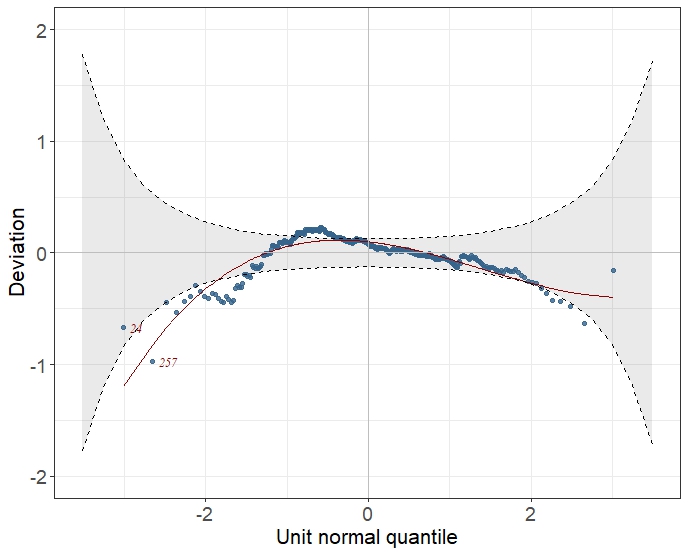}}
	\caption{ {\color{black} Worm plots  of the quantile residuals under the fitted regression models.} }
	\label{diagplots_M11}
\end{figure}

\begin{figure}[!htb]
	\centering
	\subfigure[Q-Q boxplot for RTDED model]{\includegraphics[width=6cm,height=5cm]{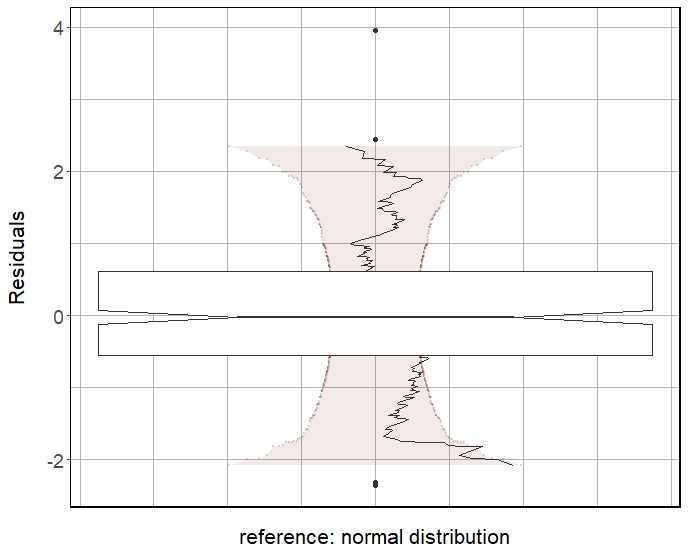}}
    \qquad
	\subfigure[Q-Q boxplot for RDTWD model]{\includegraphics[width=6cm,height=5cm]{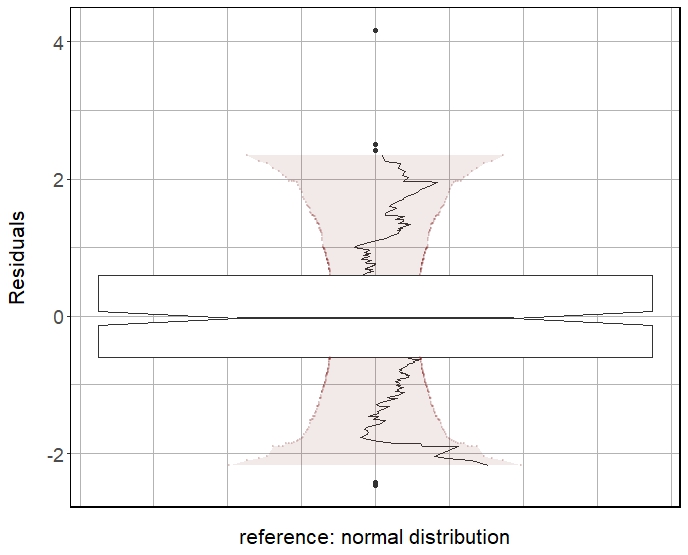}}\\
	\subfigure[Q-Q boxplot for Gamma model]{\includegraphics[width=6cm,height=5cm]{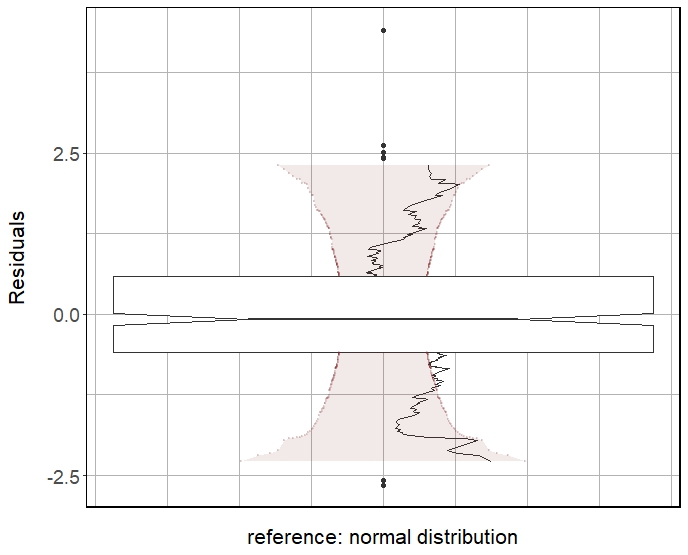}}
	\qquad
	\subfigure[Q-Q boxplot for IG model]{\includegraphics[width=6cm,height=5cm]{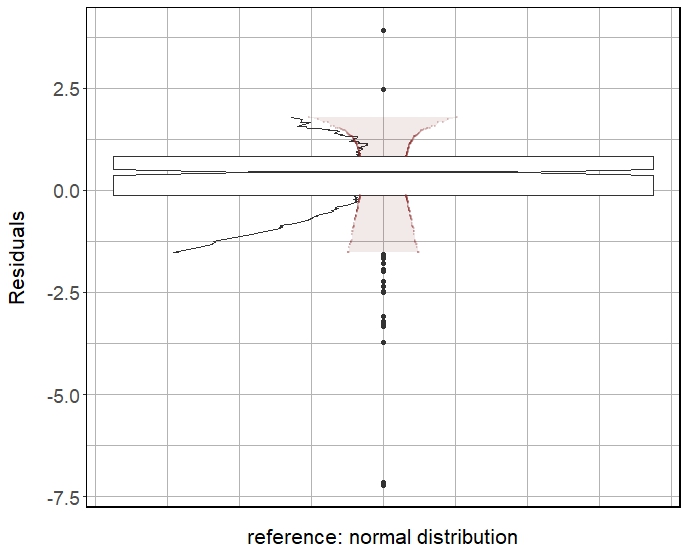}}\\
    \subfigure[Q-Q boxplot for Weibull model]{\includegraphics[width=6cm,height=5cm]{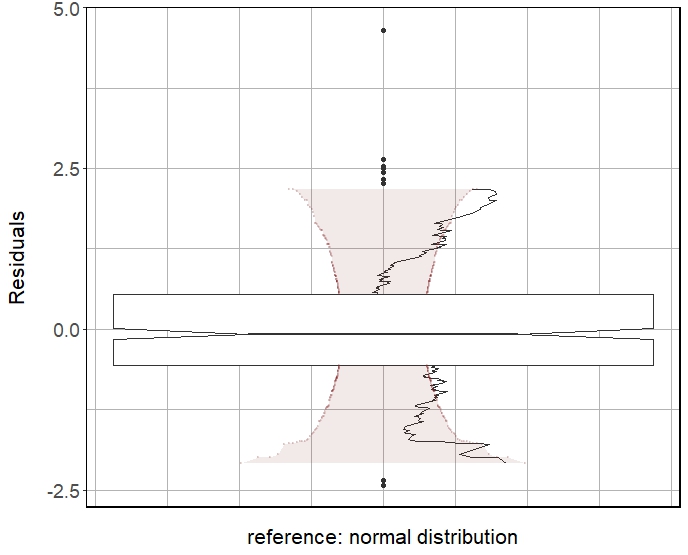}}
	\qquad
	\subfigure[Q-Q boxplot for Log-normal model]{\includegraphics[width=6cm,height=5cm]{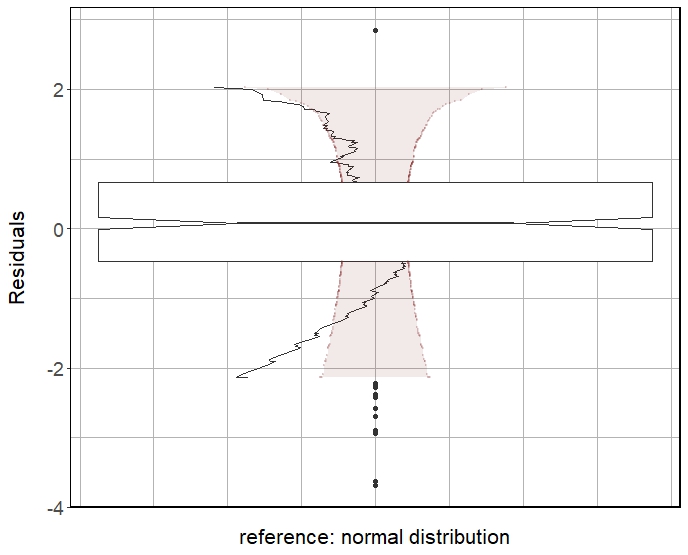}}
	\caption{{\color{black} Q-Q boxplots of the quantile residuals under the fitted regression models.} }
	\label{diagplots_M12}
\end{figure}

The worm plots for the inverse Gaussian and log-normal models reveal a lack of fit, as the points deviate significantly from the middle horizontal line, with several points falling outside the confidence intervals, indicating a systematic departure.
On the other hand, the worm plots for the RDTED, RDTWD, gamma, and Weibull models exhibit better behavior, as the points are close to the middle horizontal line. It is also evident that, in these four models, more than 95\% of the points lie between the upper and lower dotted curves, with no apparent systematic deviation. In particular, a high deviation (observation 23) is observed in all fits; however, this deviation is slightly more pronounced in the Weibull and gamma models than in the RDTED and RDTWD models.
Finally, analyzing the Q-Q boxplots in Figure \ref{diagplots_M12}, most of the fitted models exhibit some deviations from the expected behavior, displaying whiskers that do not lie within the confidence bands. However, the Q-Q boxplots for the RDTED and RDTWD models exhibit appropriate behavior, showing symmetry around the median and whiskers that clearly lie within the confidence bands. These results, obtained from the residual plots, are consistent with the AIC and BIC values presented in Table \ref{tab:aicbic}. The models with the smallest AIC and BIC values are RDTED, RDTWD, gamma, and Weibull, with the criterion values for the RDTED model being slightly smaller.

}

{\color{black} 
Table \ref{tab:estimates} shows the estimates, standard errors, $z$-statistics, and $p$-values of the test of nullity of coefficients for the RDTED, RDTWD, gamma, and Weibull regression models fitted to the foliage data. We note that all fitted models indicate the regression coefficients are statistically significant at the nominal level of 5\%. This suggests that the covariates origin and age are relevant for explaining the median foliage biomass. Additionally, under the RDTWD model, the $p$-value associated with the null hypothesis $H_0$: $\xi_1=0$, which is equivalent to $H_0$: $\nu_i=1, \forall i$, is $0.217$. This indicates that the RDTED model adequately captures the foliage data without requiring the more complex RDTWD model.}
\begin{table}[!htb]
	\centering
	\caption{{\color{black}Estimates, standard errors,  $z$-stat and $p$-values for regression models for Foliage data.}}\vspace{0.2cm}
	\label{tab:estimates}
	\scalefont{0.80}
	\def\arraystretch{1} \begin{tabular}{rrrrrrrrrr}\hline
& \multicolumn{4}{c}{RDTED model}                                   && \multicolumn{4}{c}{RDTWD model } \\  \cmidrule{2-5}\cmidrule{7-10}
                         &  Estimate & Std. error & $z$-stat &   $p$-value  && Estimate   &Std. error &$z$-stat   &$p$-value \\ \cmidrule{2-5}\cmidrule{7-10}
\text{$\mu$ submodel}    &&&&\\
		$\alpha_1$        & $-1.578$  &  $0.131$  &  $-12.041$  &  $<0.001$  && $-1.568$  &  $0.128$  &  $-12.218$  &  $<0.001$  \\
		$\alpha_2$        & $0.035$   &  $0.002$  &  $16.612$   &  $<0.001$  && $0.035$   &  $0.002$  &  $17.164$   &  $<0.001$   \\
		$\alpha_3$        & $-0.402$  &  $0.097$  &  $-4.160$   &  $<0.001$  && $-0.391$  &  $0.095$  &   $-4.132$  &  $<0.001$    \\
		$\alpha_4$        & $0.492$   &  $0.132$  &  $3.732$    &  $<0.001$  && $0.489$   &  $0.127$  &   $3.846$   &  $<0.001$   \\\cmidrule{2-5}\cmidrule{7-10}
   \text{$\sigma$ submodel}    &&&&\\
        $\gamma_1$          & $-1.904$  &  $0.327$  &  $-5.818$   &  $<0.001$  && $-2.281$  &  $0.515$  &   $-4.428$  &  $<0.001$ \\ \cmidrule{2-5}\cmidrule{7-10}
     \text{$\nu$ submodel}    &&&&\\
        $\xi_1$             &     -     &      -    &       -     &       -    && $0.061$  &  $0.049$   &   $1.237$   &  $0.217$ \\
  \cmidrule{2-5}\cmidrule{7-10}\\\hline
& \multicolumn{4}{c}{Gamma model}                                   && \multicolumn{4}{c}{Weibull model } \\  \cmidrule{2-5}\cmidrule{7-10}
                         &  Estimate & Std. error & $z$-stat &   $p$-value  && Estimate   &Std. error &$z$-stat   &$p$-value \\ \cmidrule{2-5}\cmidrule{7-10}
\text{$\mu$ submodel}    &&&&\\
		$\alpha_1$        & $-1.330$  &  $0.122$  &  $-10.865$  &  $<0.001$  &&  $-1.325$  &  $0.124$  &  $-10.683$  &  $<0.001$\\
		$\alpha_2$        & $0.035$   &  $0.002$  &  $17.372$   &  $<0.001$  &&  $0.035$   &  $0.002$  &  $16.978$   &  $<0.001$\\
		$\alpha_3$        & $-0.353$  &  $0.092$  &  $-3.828$   &  $<0.001$  &&  $-0.318$  &  $0.092$  &  $-3.436$   &  $<0.001$ \\
		$\alpha_4$        & $0.522$   &  $0.122$  &  $4.288$    &  $<0.001$  &&  $0.545$   &  $0.122$  &  $4.478$    &  $<0.001$  \\\cmidrule{2-5}\cmidrule{7-10}
  \text{$\sigma$ submodel}    &&&&\\
      $\gamma_1$          & $-0.209$  &  $0.033$  &  $-6.363$   &  $<0.001$  &&  $0.210$   &  $0.038$  &  $5.540$    &  $<0.001$   \\
  \cmidrule{2-5}\cmidrule{7-10}\\\hline
	\end{tabular}
\end{table}

{\color{black}
It is of interest to investigate whether the goodness of fit can be improved by modeling the second parameter, \(\sigma\), in all the considered models; that is, by considering double regression models. Therefore, the regression models considered now assume the same specification for the median or mean \(\mu\); however, the parameter \(\sigma\) is specified by $g_\sigma(\sigma_i) = \gamma_2{\rm age}_i$. For reasons of goodness of fit, the intercept and the covariate origin are not included in the $\sigma$ submodel. Additionally, the link function \( g_\sigma(\cdot) \) was selected to improve the fit of each regression model. Specifically, for RDTED and RDTWD models, the link function for \(\sigma\) was chosen as the square root, \( g_\sigma(\sigma) = \sqrt{\sigma} \). For the remaining models, the logarithmic link function \( g_\sigma(\sigma) = \log(\sigma) \) was used.}
\begin{figure}[!htb]
	\centering
	\subfigure[Worm plot for RDTED model]{\includegraphics[width=6cm,height=5cm]{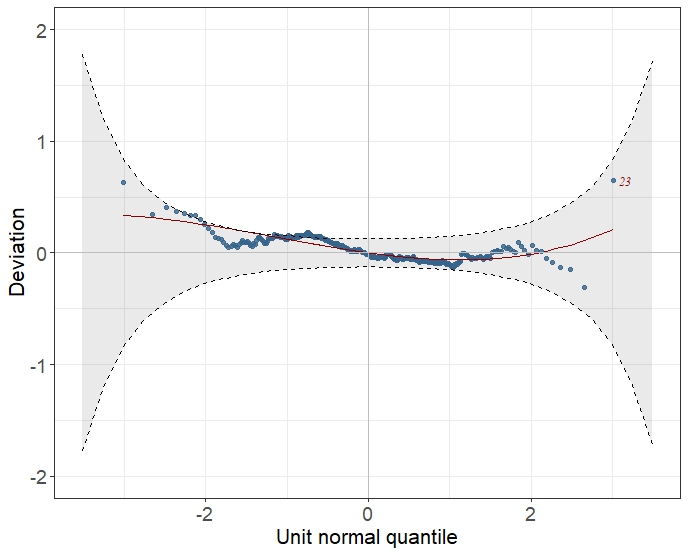}}
    \qquad
    \subfigure[Worm plot for RDTWD model]{\includegraphics[width=6cm,height=5cm]{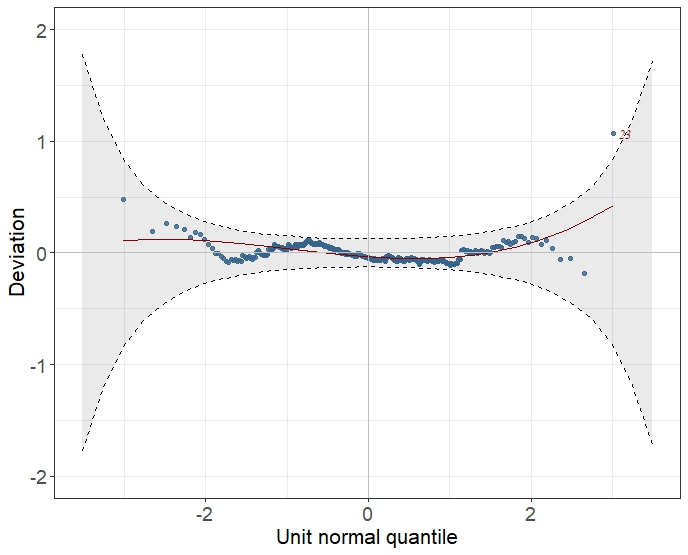}}\\
    \subfigure[Worm plot for Gamma model]{\includegraphics[width=6cm,height=5cm]{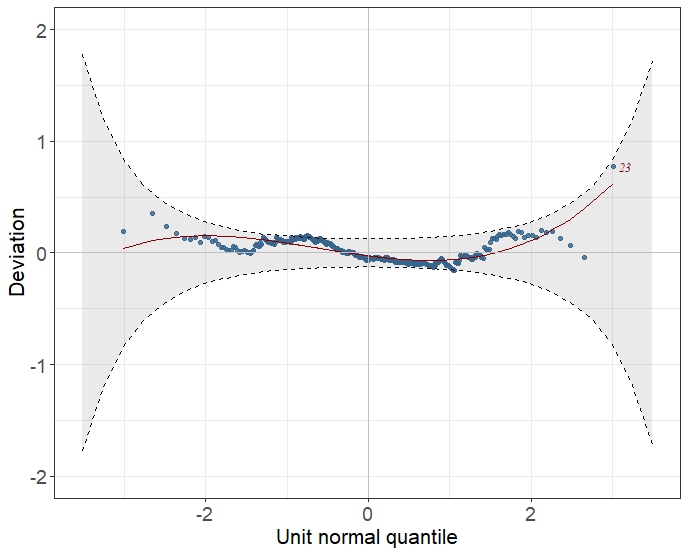}}
     \qquad
    \subfigure[Worm plot for IG model]{\includegraphics[width=6cm,height=5cm]{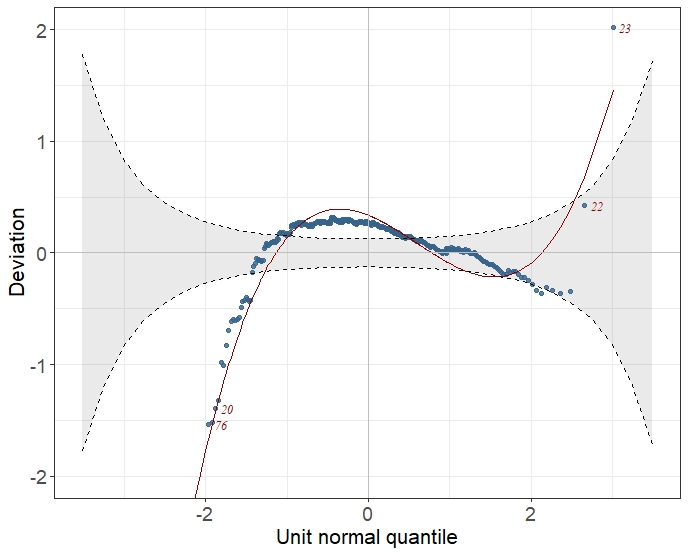}}\\
    \subfigure[Worm plot for Weibull model]{\includegraphics[width=6cm,height=5cm]{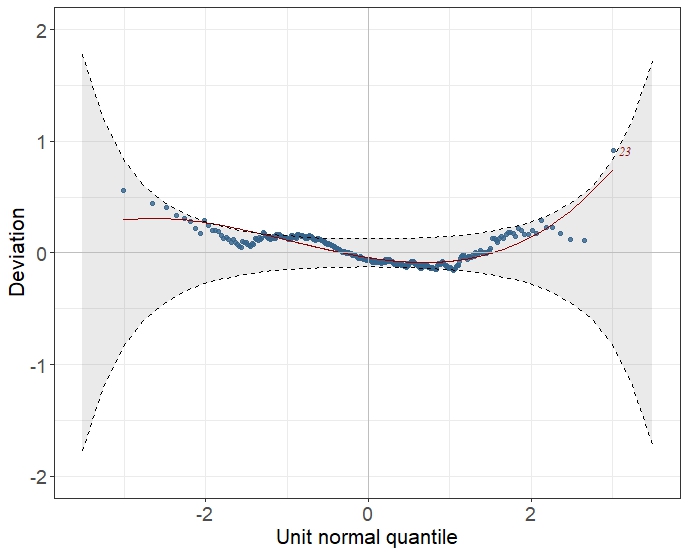}}
     \qquad
    \subfigure[Worm plot for Log-normal model]{\includegraphics[width=6cm,height=5cm]{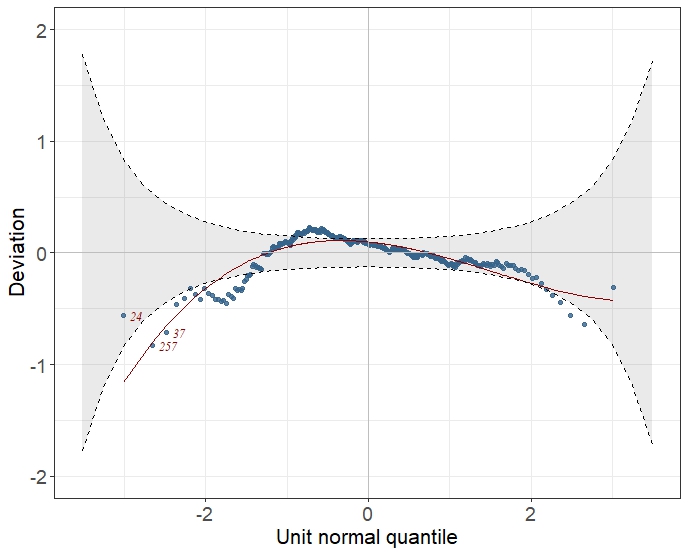}}
	\caption{ {\color{black} Worm plots  of the quantile residuals under the fitted double regression models.} }
	\label{diagplots_M21}
\end{figure}

\begin{figure}[!htb]
	\centering
	\subfigure[Q-Q boxplot for RTDED model]{\includegraphics[width=6cm,height=5cm]{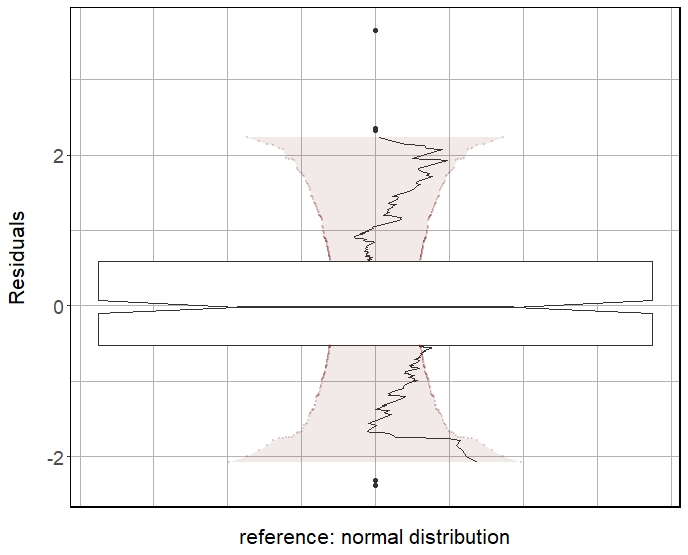}}
    \qquad
    \subfigure[Q-Q boxplot for RTDWD model]{\includegraphics[width=6cm,height=5cm]{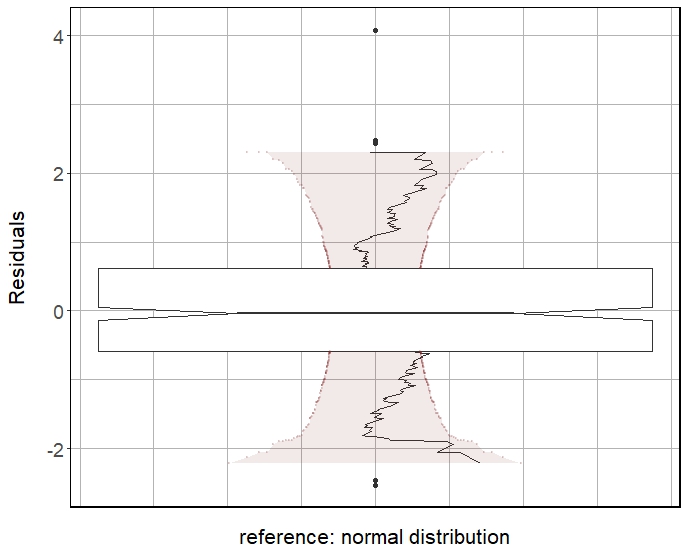}}\\
	\subfigure[Q-Q boxplot for Gamma model]{\includegraphics[width=6cm,height=5cm]{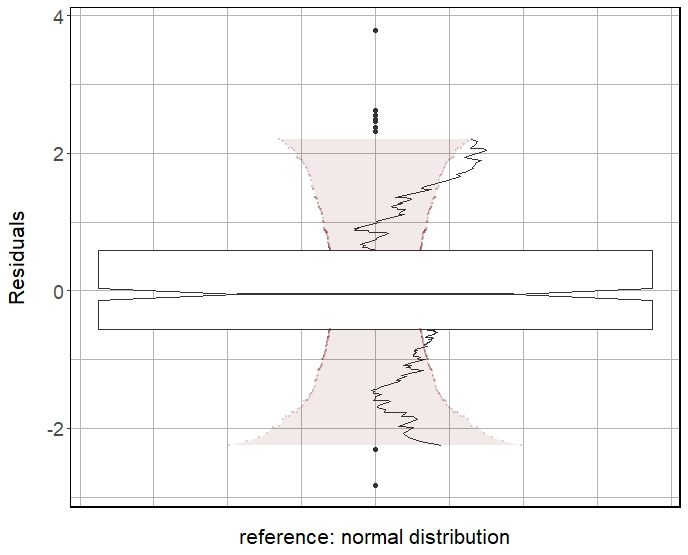}}
    \qquad
	\subfigure[Q-Q boxplot for IG model]{\includegraphics[width=6cm,height=5cm]{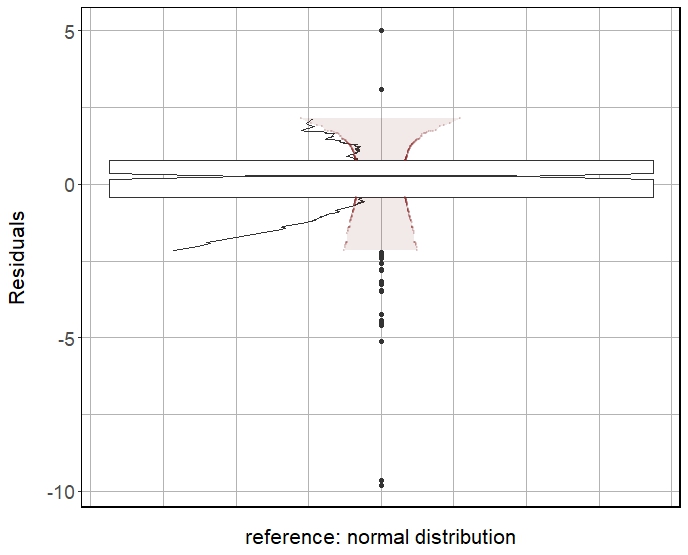}}
    \qquad
    \subfigure[Q-Q boxplot for Weibull model]{\includegraphics[width=6cm,height=5cm]{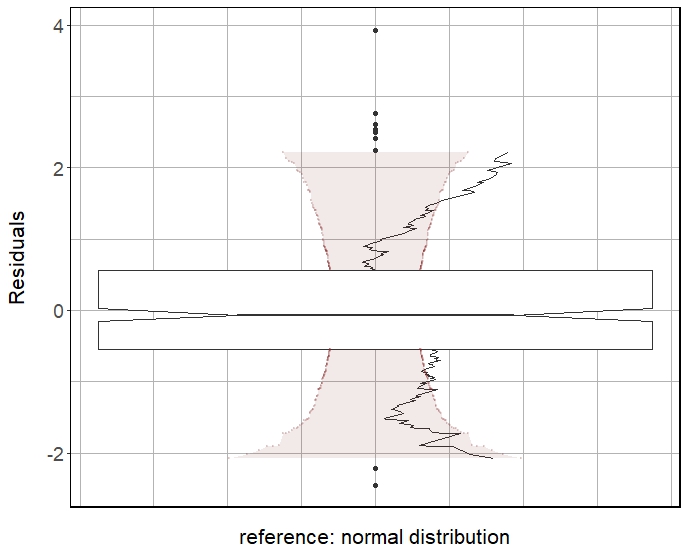}}
    \qquad
	\subfigure[Q-Q boxplot for Log-normal model]{\includegraphics[width=6cm,height=5cm]{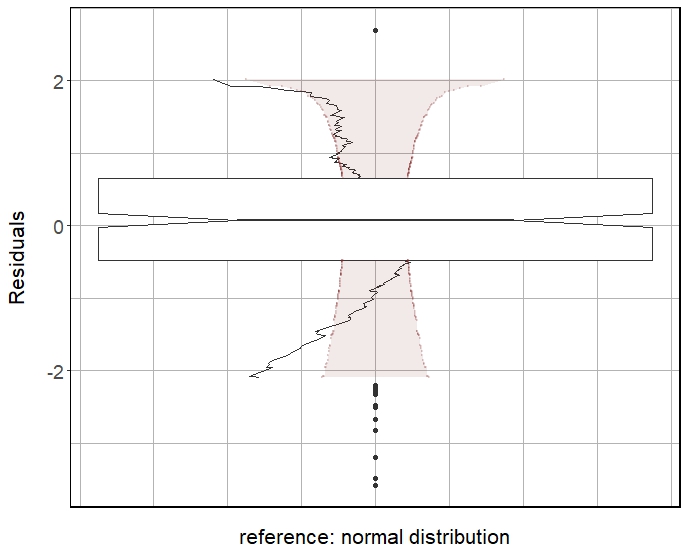}}
	\caption{{\color{black} Q-Q boxplots of the quantile residuals under the fitted double regression models.} }
	\label{diagplots_M22}
\end{figure}

{\color{black} Figures \ref{diagplots_M21} and \ref{diagplots_M22} show the worm plots and Q-Q boxplots for the fitted double regression models. Once again, there is a clear lack of fit for the inverse Gaussian and log-normal models, as the points deviate significantly from the middle horizontal line, with several points falling outside the confidence intervals.
On the other hand, the worm plots for the RDTED, RDTWD, gamma, and Weibull models once again demonstrate similarly good behavior. In particular, the worm plot under the RDTWD model shows slightly better behavior. Also, the Q-Q boxplots shown in Figure \ref{diagplots_M22} reveal some deviations from the expected behavior, with whiskers extending beyond the confidence bands for the gamma and Weibull models. In contrast, the Q-Q boxplots for the RDTED and RDTWD models exhibit appropriate behavior, characterized by symmetry around the median and whiskers that lie within the confidence bands. These findings are consistent with the AIC and BIC values presented in Table \ref{tab:aicbic2}, where the RDTED and RDTWD models have the lowest AIC and BIC values. Additionally, we observe that the RDTED, gamma, and inverse Gaussian models show slight increases in AIC and BIC values when the covariate age is added to the $\sigma$ submodel, whereas the RDTWD, Weibull, and log-normal models exhibit slight decreases in these criteria.}
\begin{table}[!htb]
\centering
\caption{{\color{black} AIC and BIC values for the fitted double regression models for foliage data.}}
\label{tab:aicbic2}
\begin{tabular}{l|cc}
\toprule
\textbf{Model} & \textbf{AIC} & \textbf{BIC} \\\hline
RDTED                     &  996.98  & 1016.75\\
RDTWD                     &  994.83  & 1018.55\\
Gamma                     &  1000.92 & 1020.69 \\
Inverse Gaussian          &  1367.82 & 1387.58\\
Weibull                   &  1006.50 & 1026.27 \\
Log-normal                &  1015.34 & 1035.11\\
\bottomrule
\end{tabular}
\end{table}

{\color{black} 
Table \ref{tab:estimates2} shows the estimates, standard errors, $z$-statistics, and $p$-values of the test of nullity of coefficients for the double RDTED and RDTWD regression models fitted to the foliage data. We observe that the estimates for the $\mu$ submodel are very close to those shown in Table \ref{tab:estimates}, with the same significance level. This finding is consistent with the lack of improvement in goodness of fit for the RDTED model and the slight improvement observed in the RDTWD model, as indicated by the residual plots and criterion values.
However, under the RDTWD model, the $p$-value associated with the null hypothesis $H_0$: $\xi_1=0$ is 0.038. Therefore, at a nominal significance level of 5\%, \(\nu_i\) is not equal to 1.0 for all \(i\); that is, the fits obtained by RDTED and RDTWD differ considerably for the foliage data when \(\sigma\) is also modeled by the covariate age. Based on the residual plots, AIC, and the results of the Wald test, the best-fitting model for the foliage data, considering double regression, is the RDTWD model.

 }
\begin{table}[!htb]
	\centering
	\caption{{\color{black}Estimates, standard errors,  $z$-stat and $p$-values for double regression models for foliage data.}}\vspace{0.2cm}
	\label{tab:estimates2}
	\scalefont{0.80}
	\def\arraystretch{1} \begin{tabular}{rrrrrrrrrr}\hline
& \multicolumn{4}{c}{RDTED model}                                   && \multicolumn{4}{c}{RDTWD model } \\  \cmidrule{2-5}\cmidrule{7-10}
                         &  Estimate & Std. error & $z$-stat &   $p$-value  && Estimate   &Std. error &$z$-stat   &$p$-value \\ \cmidrule{2-5}\cmidrule{7-10}
\text{$\mu$ submodel}    &&&&\\
		$\alpha_1$        & $-1.676$  &  $0.142$  &  $-11.789$  &  $<0.001$  && $-1.618$  &  $0.134$  &  $-12.070$  &   $<0.001$    \\
		$\alpha_2$        & $0.037$   &  $0.002$  &  $15.915$   &  $<0.001$  && $0.036$   &  $0.002$  &  $16.777$   &   $<0.001$  \\
		$\alpha_3$        & $-0.373$  &  $0.100$  &  $-3.735$   &  $<0.001$  && $-0.364$  &  $0.094$  &  $-3.884$   &   $<0.001$    \\
		$\alpha_4$        & $0.510$   &  $0.136$  &  $3.751$    &  $<0.001$  && $0.503$   &  $0.126$  &  $3.999$    &   $<0.001$  \\\cmidrule{2-5}\cmidrule{7-10}
        {\it $\sigma$} \text{submodel}    &&&\\
        $\gamma_2$        & $0.007$   &  $0.001$  &  $4.810$    &  $<0.001$  && $0.005$   &  $0.001$  &   $3.779$   &  $<0.001$ \\
  \cmidrule{2-5}\cmidrule{7-10}\\
  {\it $\nu$} \text{submodel}    &&&\\
        $\xi_1$          &     -      &     -     &      -      &       -    && $0.092$   &  $0.044$  &   $2.078$   &  $0.038$     \\\cmidrule{2-5}\cmidrule{7-10}\\\hline
	\end{tabular}
\end{table}

{\color{black} 

To evaluate whether a more complex RTDWD model could provide a better fit, we also considered introducing covariates in the $\nu$ submodel as $ \log(\nu_i) = \xi_1 + \xi_2{\rm originN}_i + \xi_3{\rm originP}_i$. Table \ref{tab:estimates3} shows the estimates, standard errors, $z$-statistics, and $p$-values of the test of nullity of coefficients for the RTDWD regression model with three submodels. At a nominal level of 5\%, all regression coefficients are significant. In particular, the covariate origin is relevant for modeling the parameter \(\nu\). Furthermore, the AIC and BIC values for this fitted RTDWD model are 991.78 and 1023.41, respectively. Thus, while the AIC is smaller compared to previous models, the BIC is larger. The worm plot and Q-Q boxplot for this fitted model (not shown) are similar to those obtained previously for RDTWD models. Overall, we do not observe a substantial improvement in the RDTWD fit. Therefore, we selected the RTDWD model with covariates in the median and \(\sigma\) submodels as the final model.
\begin{table}[!htb]
	\centering
	\caption{{\color{black}Estimates, standard errors,  $z$-stat and $p$-values for the RDTWD regression model with covariates in all submodels for foliage data.}}\vspace{0.2cm}
	\label{tab:estimates3}
	\scalefont{0.80}
	\def\arraystretch{1} \begin{tabular}{rrrrr}\hline
                         &  Estimate & Std. error & $z$-stat &   $p$-value  \\ \cmidrule{2-5}
\text{$\mu$ submodel}    &&&&\\
		$\alpha_1$        &  $-1.564$  &  $0.126$  &  $-12.412$  & $<0.001$   \\
		$\alpha_2$        &  $0.036$   &  $0.002$  &  $17.399$   & $<0.001$\\
		$\alpha_3$        &  $-0.455$  &  $0.097$  &  $-4.702$   & $<0.001$  \\
		$\alpha_4$        &  $0.390$   &  $0.141$  &  $2.757$    & $0.006$ \\\cmidrule{2-5}
        {\it $\sigma$} \text{submodel}    &&&\\
        $\gamma_2$        & $ 0.005$   &  $0.001$  &  $3.650$    &  $<0.001$\\
  \cmidrule{2-5}\\
  {\it $\nu$} \text{submodel}    &&&\\
        $\xi_1$          &  $0.247$    &  $0.071$  &  $3.481$    &  $<0.001$     \\
        $\xi_2$          &  $-0.202$   &  $0.083$  &  $-2.422$   &  $0.016$    \\
        $\xi_3$          &  $-0.241$   &  $0.113$  &  $-2.134$   &  $0.033$     \\
        \cmidrule{2-5}\\\hline
	\end{tabular}
\end{table}

}

{\color{black}

Under the double RDTWD regression model, we have the following conclusions.
We estimate that a 10-year increase in tree age results in a 43.3\% increase in the median foliage biomass, since that $\exp(0.036\times 10) = 1.433$. The median foliage biomass of trees originating from coppice decreases by 30.5\% compared to trees from natural sources, and increases by 65.4\% compared to planted trees, as indicated by \(\exp(-0.364) = 0.695\) and \(\exp(0.503) = 1.654\), respectively.

\section{Concluding remarks}\label{CR}


This paper introduced a new generator of probability distributions, termed the Dutta transformed-G (DT-G) family, which extends a baseline distribution through the inclusion of an additional tail-modulation parameter. The proposed transformation is simple, easy to implement, and capable of generating a broad range of density and hazard rate shapes while preserving the overall structure of the baseline model. An important feature of the DT-G family is that the additional parameter admits a clear interpretation in terms of tail behavior, allowing the construction of heavier-tailed models when needed.

Several theoretical properties of the DT-G family were derived, including results related to critical points, modality, stochastic representation, identifiability, quantiles, moments, and truncated moments. These findings provide a comprehensive understanding of the mathematical structure of the proposed family and establish a foundation for further methodological developments.

To demonstrate the practical usefulness of the DT approach, two particular members of the family, namely the Dutta transformed exponential distribution (DTED) and the Dutta transformed Weibull distribution (DTWD), were studied in detail. Alternative parameterizations were developed so that one of the model parameters corresponds directly to the median of the response variable. This feature enhances interpretability in regression settings, particularly when dealing with asymmetric data or data containing extreme observations.

Building upon these parameterizations, two new regression models, referred to as the RDTED and RDTWD models, were proposed within the GAMLSS framework. Maximum likelihood estimation procedures and diagnostic tools were discussed, providing a complete inferential framework for practical applications. The real-data application involving foliage biomass measurements of small-leaved lime trees illustrated the flexibility and effectiveness of the proposed models. The results showed that the RDTED and RDTWD models provide competitive and interpretable alternatives to existing regression models commonly used in the literature.

Overall, the DT methodology offers a versatile mechanism for constructing new distributions and regression models with interpretable parameters and flexible tail behavior. Future research may explore other baseline distributions within the DT-G family, multivariate extensions, censored-data applications, and more general regression structures. To facilitate reproducibility and further developments, all computational routines used in this study are freely available at
\url{https://github.com/terezinharibeiro/DT_Regression_Models}.
}

	\paragraph*{Disclosure statement}
There are no conflicts of interest to disclose.

{\color{black}

\paragraph*{Acknowledgements} The authors thank the reviewers for their constructive suggestions and valuable comments on a previous version of this manuscript.

  }

\bibliographystyle{apalike}
\bibliography{ref_WTG}

\end{document}